\documentclass[12pt]{article}
\usepackage{amsfonts}
\usepackage{eurosym}
\usepackage{float}
\usepackage{makeidx}
\usepackage{graphicx}
\usepackage{amsmath}
\usepackage[dcucite]{harvard}

\newtheorem{Theo}{Theorem}

\newtheorem{Corol}{Corollary}
\newtheorem{Def}{Definition}
\newtheorem{Le}{Lemma}

\newtheorem{Rem}{Remark}

\input{tcilatex}
\begin{document}

\date{}
\title{\textbf{A study on fault diagnosis in nonlinear dynamic systems with
uncertainties}}
\author{{\normalsize Steven X. Ding} \\
{\normalsize Institute for Automatic Control and Complex Systems (AKS)}\\
{\normalsize University of Duisburg-Essen, 47057 Duisburg, Germany} \and 
{\normalsize Linlin Li} \\
{\normalsize School of Automation and Electrical Engineering}\\
{\normalsize University of Science and Technology Beijing, Beijing 100083,
China}}
\maketitle




\bigskip

\textbf{Abstract}: In this draft, fault diagnosis in nonlinear dynamic
systems is addressed. The objective of this work is to establish a
framework, in which not only model-based but also data-driven and machine
learning based fault diagnosis strategies can be uniformly handled. To this
end, a paradigm is followed, which is different from the control
theoretically oriented model-based framework. Instead of the
well-established input-output and the associated state space models, stable
image and kernel representations are adopted in our work as the basic
process model forms. The idea behind this handling is to deal with fault
diagnosis issues in the process input-output data space, along the line of
data-driven and machine learning methods. Using the image and kernel
representations the nominal system dynamics can then be modelled as a
lower-dimensional manifold embedded in the process data space. To achieve a
reliable fault detection as a classification problem, projection technique
is a capable tool that projects the process data onto the image manifold
towards classifying the nominal dynamics and faulty/uncertain dynamics. For
nonlinear dynamic systems, we propose to construct projection systems in the
well-established framework of Hamiltonian systems and by means of the
normalised image and kernel representations. It is known that, in case of
linear systems, the norm-based evaluation of the residual signal, created
from the difference between the data and its projection, is widely used in
fault detection. Thanks to the Hilbert Projection Theorem and the
Pythagorean equation, the norm-based evaluation can uniquely distinguish
nominal dynamics from the faulty/uncertain dynamics. For nonlinear dynamic
systems, process data form a non-Euclidean space. Consequently, the
norm-based distance defined in Hilbert space is not suitable to measure the
distance from a data vector to the manifold of the nominal dynamics. To deal
with this issue, we propose to use a Bregman divergence, a measure of
difference between two points in a space, as a solution. Moreover, for our
purpose of achieving a performance-oriented fault detection, the Bregman
divergences adopted in our work are defined by Hamiltonian functions. A
distinguishing contribution of our work is the development of a scheme that
combines the Hamiltonian systems, their Legendre dual systems and Bregman
divergences induced by the corresponding Hamiltonian functions. This scheme
not only enables to realise the performance-oriented fault detection, but
also uncovers the information geometric aspect of our work. Specifically,
the Bregman divergence together with Hamiltonian Legendre dual systems may
induce a dually flat Riemannian manifold in the data space that is regarded
as a dualistic extension of the Euclidean space. In this context, the
Hamiltonian system based projection can be interpreted as a geodesic
projection. The last part of our work is devoted to the kernel
representation based fault detection and uncertainty estimation that can be
equivalently used for fault estimation. To this end, the concepts of the
uncertainty model and the manifold of uncertainty data are introduced and
used to model uncertainties (including faults) corrupted in the process
data. It is demonstrated that the projection onto the manifold of
uncertainty data, together with the correspondingly defined Bregman
divergence, is also capable for fault detection. In particular, it is proved
that such a projection is a least squares estimation of the uncertainty
corrupted in data and subject to the uncertainty model. Hence, the kernel
representation based projection can serve as an optimal estimator for
variations in the process data caused by faults.

\bigskip


\section{Introduction}

Having undergone a rapid development over a couple of decades, model-, in
particular, observer-based fault diagnosis technique has become well
established as an active research area in control theory and engineering 
\cite%
{Frank90,DingJPC97,VRK03-I,MDES2009,HKKS2010survey,GCD-2015survey,ZXD_2018}.
Technically speaking, an observer-based fault detection system consists of
two major units, a residual generator and evaluator. For their design and
analysis system input-output and the associated state space models build a
common basis, and advanced control theoretical methods serve as the major
theoretical tool \cite{PFC89,Gertler98,CP99,Blanke06,Ding2008}.
Specifically, the observer-based residual generator design is mainly
achieved by means of advanced observer theory, while the construction of the
evaluator with a focus on the threshold setting is performed on the basis of
norm-based system analysis \cite{Ding2008,Ding2020}. To put it in a
nutshell, state of the art of observer-based fault diagnosis technique is
marked by the well-established framework for fault diagnosis in linear time
invariant (LTI) systems, while a systematic dealing with fault diagnosis
issues is missing for nonlinear systems.

\bigskip

In their survey \cite{AlcortaCEP97}, Alcorta-Garcia and Frank reviewed state
of the art of nonlinear observer-based fault detection in the 90's. The
major methods include the application of feedback-based linearisation and
differential algebra techniques to observer-based residual generator design 
\cite{SeligerCDC91,HammouriAC99,KH_2001}, and the geometric approach to
nonlinear fault detection \cite{Persis_2001}. Later, considerable attention
has been paid to the application of those special techniques that enable
dealing with analysis and synthesis of nonlinear dynamic systems more
efficiently. These are, amongst others, linear matrix inequalities (LMI)
technique to address a special class of nonlinear systems with, for
instance, Lipschitz nonlinearity \cite{PMZ_2007,ZPP_automatica_2010} or
sector bounded nonlinearity \cite{He09}, fuzzy technique based fault
detection \cite{NSD_2007,CAD_2013_Automatica,Jiang2014,LDQYZ_2016,LDQYX_2017}%
, adaptive fault diagnosis for nonlinear systems \cite%
{XZ_automatica_2004,ZPP_2010,Jiang2014}, linear parameter-varying fault
detection systems \cite{BB2004,ACM_LPV_2009}, and sliding mode
observer-based fault detection {\cite{FBPD_2004,YE_2008}. Although all these
efforts are successful in dealing with fault detection of the corresponding
system classes, a common theoretical framework for design and analysis of
nonlinear fault detection systems seems missing. In our early efforts
towards such a framework, }we noticed that little attention was paid to
essential issues {of nonlinear observer-based }fault detection systems like {%
existence conditions, detectability and parameterisation. }With the aid of
nonlinear system theory as a major tool \cite%
{Vidyasagar80,Sontag_Wang_1995,LU1995_nonlinear,vanderSchaftbook}, {these
issues have been examined and the first results were reported in \cite%
{YDL_SCL_2015,YDL-2016,Ding2020}. }

\bigskip

Fault detection is a classification problem, namely classifying those
potential anomalies from the collected data during system operations. The
use of a system model and, based on it, an observer allows to distinguish
anomalies from the normal input-output dynamics and thus leads to an
efficient classification. On the other hand, in industrial practice, the use
of a mathematical model is generally associated with model uncertainties.
Coupled with external disturbances (unknown inputs), they considerably
complicate a reliable classification, in particular when nonlinear dynamic
systems are under consideration. In fact, distinguishing anomalies/faults
from the uncertainties existing during fault-free operations is a most
crucial issue in dealing with classification/fault detection. A further
challenging issue is to detect process faults, also called component faults
in the literature \cite{Frank90,Ding2008}. As reported in innumerable
publications, observer-based fault detection methods are capable to detect
sensor and actuator faults in dynamic systems. Two plausible reasons for
such successful applications are, (i) an input-output model gives a natural
modelling of sensor and actuator faults, (ii) the norm-based evaluation of
residual signals (the difference of the measurement and its estimate) builds
an explicit indicator for these two types of faults. Differently, process
faults are corrupted with model uncertainties, and as a consequence of
nonlinear system dynamics, process data build a manifold that is a
non-Euclidean space. It is known that usual Euclidean distances or inner
products may not be appropriate to measure the dissimilarity between the
data points on a manifold, which is the core purpose of classification. That
process faults as well as uncertainties can result in remarkable system
performance degradation, but are hardly detected using norm-based evaluation
function is the further aspects that motivate our work.

\bigskip

In the era of artificial intelligence and big data, machine learning (ML)
based classification methods have become for a long time now the mainstream
in the research area of engineering fault diagnosis. Classification-based
fault diagnosis is a technique that uses ML-classifiers to detect faults in
a system based on the features extracted from the process data. As a crucial
step towards reliable and explainable fault diagnosis, dimensionality
reduction technique, also known as manifold learning, serves as a capable
tool, where the irrelevant and redundant informations in the data are
reduced, leading to not only lower computation complexity but also, and more
importantly, better classification and fault diagnosis performance. In order
to capture nonlinear and nonstationary characteristics of the data,
nonlinear dimensionality reduction methods have been developed, among them,
the so-called Isomap (Isometric mapping) \cite{Isomap2000} and LLE
(Locally-linear embedding) \cite{LLE2000} are two representative ones.
Informally speaking, an essential assumption in manifold learning is that
the data have been generated from a lower-dimensional manifold that is
embedded inside of a higher-dimensional space composed of redundant and
irrelevant information (uncertainties). The basic idea of Isomap is the use
of a geodesic distance, which enables to incorporate the manifold structure
in the resulting embedding. The objective of the LLE algorithm is to create
a data mapping from the higher-dimensional manifold to the lower-dimensional
one and preserving local neighbourhoods at every point of the underlying
manifold. In recent years, the so-called autoencoder (AE) technique, a
capable ML-based algorithm for dimensionality reduction and feature learning 
\cite{Hinton2006,Goodfellow2016}, is widely and successfully applied to
fault diagnosis and anomaly detection \cite%
{JIANG2017,Ahmed2021,HZP-CEP-2022,LYY-NCA-2021,RZZ-TII-2020}. An AE consists
of an encoder and a decoder, where the encoder compresses the data in the
higher-dimensional space into a lower-dimensional feature manifold
represented by the so-called latent variable, and the decoder is driven by
the latent variable and delivers the process variables reflecting, in case
of fault diagnosis, nominal process operations. The encoder and decoder are
constructed using neural networks (NNs) that are learnt in such a way that
the latent variable (the lower-dimensional feature) contains exclusive
informations of the nominal process dynamics, from which the process
variables in the nominal state are then fully reconstructed. To put it in a
nutshell, dimensionality reduction is a capable strategy to remove redundant
and irrelevant informations in the data and thus recover exclusive
information about the nominal process operations by means of projecting the
data in the higher-dimensional space to the lower-dimensional feature
manifold.

\bigskip

A comparison of the observer-based and dimensionality reduction-based ML
diagnose methods reveals the following insightful differences of these fault
detection strategies, in spite of the evidently different assumptions on the
types of available process knowledge and information,

\begin{itemize}
\item the input-output/state space models describing the nominal system
dynamics build the system core of observer-based fault detection, while the
ML-based methods handle the nominal system dynamics as a sub-manifold
(feature manifold) embedded in the (higher-dimensional) data space composed
of all the process data with uncertainties,

\item accordingly, the observer-based methods handle system uncertainties by
means of residual generation, and the ML-based methods make use of
projection techniques,

\item and consequently, in the framework of the observer-based fault
detection the detection decision is made on the basis of a norm-based
residual evaluation, while for the ML-based methods, the decision is\ based
on the distance metric between the current data vector and the
lower-dimensional feature manifold. It is noteworthy that in non-Euclidean
space the concept of a geodesic distance is often applied in order to
capture the nonlinear structure of the data manifold.
\end{itemize}

In a short form, these major differences are summarised as%
\begin{eqnarray*}
&&\left\langle \text{I/O-model, residual, norm-based evaluation}%
\right\rangle \text{ }vs. \\
&&\left\langle \text{feature manifold, projection, geodesic distance}%
\right\rangle ,
\end{eqnarray*}%
whose centerpiece is handling of nominal dynamics and uncertainties.

\bigskip

Handling of model uncertainties using the projection technique in a Hilbert
subspace is nothing new in control theory. In the 90s, gap metric was widely
applied in robust control theory to measure influences of model
uncertainties on a feedback control system \cite%
{Georgiou&Smith90,Vinnicombe-book}. Roughly speaking, a gap is a similarity
metric between subspaces in Hilbert space based on an orthogonal projection 
\cite{Georgiou88,Vinnicombe-book,Feintuch_book}. Applied to robust control
theory, a gap metric gives, for instance, the "distance" between the system
dynamics with and without uncertainties. The use of gap metrices to fault
detection was proposed by \cite{Ding2015a}, and comprehensively investigated
in \cite{LD-Automatica-2020,Wang-gap-2022}. Motivated by the promising
results in these works and inspired by the aforementioned ML-based fault
detection methods, an initial project has been launched to investigate
application of orthogonal projection technique to fault detection in LTI
systems \cite{ding2022}. The primary objective of the project is to
establish a projection-based fault detection technique as a potential
alternative framework to the observer-based one and to compare the capacity
of the two detection schemes. Below are two distinct conclusions,

\begin{itemize}
\item in the projection-based fault detection framework, residual generation
and evaluation can be well realised in an optimal and integrated manner, and

\item a projection-based fault detection system can be, under certain
conditions, equivalently realised by a class of optimal observer-based fault
detection systems.
\end{itemize}

The theoretical basis for these conclusions is the orthogonal projection in
Hilbert space which allows an orthogonal decomposition of a process data
vector into an optimal estimate (projection) and the residual. As a result,
the norms of the process data, the residual vector and the projection
satisfy, following the Pythagorean Theorem, the Pythagorean equation, which
leads to an optimal fault detection.

\bigskip

Motivated and driven by the aforementioned preliminary results, the main
project has started with the objective of establishing a projection-based
fault diagnosis framework, in which nonlinear dynamic systems with
uncertainties are considered, and an optimal system performance-oriented
fault detection as well as fault estimation are defined and approached. This
is also the major objective of this draft reporting the main results of our
work. We would like to mention that the overall goal of our efforts is to
build a uniform framework, in which both control theory based and ML-based
methods can be developed for fault diagnosis in dynamic control systems. To
be specific, the design of ML-based fault diagnosis systems could be well
guided by known control-theoretic knowledge and, on the other hand, the
design and realisation of model-based fault diagnosis systems are supported
by data-driven and ML-algorithms. Our recent work on control theoretically
explainable application of autoencoder methods to fault detection \cite%
{LDLCX2022} is an example of the former case, while the latter is inspired
by current research on physics-informed machine learning (PINN) \cite%
{PINN2021}.

\bigskip

To approach the above described objective, we focus on the following four
major tasks.

\begin{itemize}
\item Definition of nominal system dynamics as a manifold in the process
data space and realisation of normalised image and kernel representations.
Let the process input and output data build a data space. It is known that,
for linear systems, the so-called (stable) image and kernel representations
(SIR and SKR) are system model forms that are equivalent to the input-output
(transfer function) and the associated state space models \cite%
{Vinnicombe-book}. In particular, the normalised SIR and SKR, as inner and
co-inner systems \cite{Hoffmann1996}, play an importance role in the
(orthogonal) projection related research \cite%
{Georgiou88,Vinnicombe-book,Feintuch_book}. By means of SIR and SKR, two
equivalent subspaces, the image and kernel subspaces, are defined in Hilbert
space. The analogue concepts are also well defined for nonlinear systems 
\cite{vanderSchaftbook}. Nevertheless, the lossless properties of inner and
co-inner (refer to Definition \ref{Def2-2}) are of special importance in our
subsequent study and motivate us to follow the works by van der Schaft and
his co-worker on nonlinear systems \cite{Scherpen1994,Ball1996,PS-AC2005}.
The core of these works is various configurations of the so-called
Hamiltonian extension \cite{Schaft-book1987} as Hamiltonian systems, which
serve as the major system theoretical tool in our work as well.

\item Design and analysis of projection systems. For LTI systems, the
orthogonal projection onto the image space is a series system composed of
the normalised SIR and its conjugate \cite{Georgiou88}. To our best
knowledge, no study has been reported about analogue results for nonlinear
systems, possibly for the reason that for nonlinear systems such a
projection is not a geodesic one and thus seems useless. Despite that, we
propose to configure the projection as a Hamiltonian system using the
normalised SIR. The centerpiece in this work is undoubtedly the Hamiltonian
function of the system, which describes the system dynamics in the sense of
energy transfer in the system \cite{Scherpen1994,Ball1996,PS-AC2005}, and is
thus interpreted as a system performance function. Associated with it, a key
concept, the so-called Legendre transform, is introduced, which gives a dual
relation of the input and output vectors of the Hamiltonian system as well
as a dual form of the Hamiltonian system \cite{PS-AC2005,Amari2016}. For our
purpose of fault detection, the dual pair of the Hamiltonian functions is
analysed corresponding to the nominal system operation and operations when
uncertainties or even faults exist. Although this task is primarily
dedicated to construct projection systems, it serves as the system-theoretic
basis for the next task as well.

\item Bregman divergence-based performance-oriented fault detection. While
the first two tasks are, more or less, extended applications of the existing
methods, this task is crucially different from the projection-based scheme
for LTI systems, where the residual is built by means of the difference of
the data vector and its projection and further norm-based evaluated. This
task is devoted to three issues: (i) handling of process data in
non-Euclidean space aiming at classification, (ii) performance-oriented
fault detection, and (iii) interpretation of the projection-based method. To
this end, Bregman divergences \cite{BREGMAN1967} are introduced. A Bregman
divergence is defined in terms of a convex function and serves as a measure
of difference between two points in a space. We propose to define Bregman
divergences in terms of Hamiltonian functions. The idea behind that is to
measure the difference between the data vector and its projection with
respect to the change in the Hamiltonian function as a performance function
under consideration. Observe that Bregman divergences form an important
class of divergences and are widely used in information geometry \cite%
{Amari2016}. In particular, a Bregman divergence together with a Legendre
transform may induce a dually flat Riemannian manifold in the data space
that may be regarded as a dualistic extension of the Euclidean space \cite%
{Amari2016,Nielsen2020}. In this context, the interpretation of the
Hamiltonian system based projection as a geodesic projection will be studied.

\item Modelling of data uncertainty, manifold of uncertain data, and fault
detection and estimation. In the projection-based fault diagnosis framework
for LTI systems, thanks to the Pythagorean Theorem and the Pythagorean
equation, projections onto the image subspace and its orthogonal complement
are dual problems and can be addressed, in a certain sense, in a uniform
manner. For process data in a non-Euclidean space, they cannot be decomposed
into two components representing nominal and uncertain/faulty data,
respectively. In other words, nominal data are corrupted with uncertain data
and they are inseparable. For our purposes, the concepts of (data)
uncertainty model and manifold of uncertain data are introduced on the basis
of SKR. While the data in the image manifold are created under nominal
system dynamics, the data in the manifold of uncertain data are assumed to
contain uncertainties exclusively. In this sense, the SKR-based uncertainty
model is a dual form of the SIR-based model for nominal dynamics. Here, we
place a greater emphasis on the possible interpretation of uncertainties as
faults. In such a system setting, Hamiltonian systems, the associated
Hamiltonian functions and the corresponding Bregman divergences, dual to the
ones defined in previous tasks, are defined, designed, and then applied for
the fault detection purpose. In addition, on the assumption that the
uncertainty is caused by faulty operations, the SKR-based projection onto
the manifold of uncertain data delivers an estimate for the faulty data. The
performance of such an estimation will finally be analysed.
\end{itemize}

The intended contributions of our work can be summarised in three levels and
are described as follows.

\begin{itemize}
\item Establishment of a fault diagnosis framework. This framework provides
us with an optimal design of fault diagnosis systems for dynamic control
systems, as an alternative to the observer-based one. More importantly, its
application is not limited to the design of model-based fault diagnosis
systems, but also useful to guide the construction of ML-based fault
diagnosis systems.

\item Novel fault diagnosis schemes and methods. On the basis of normalised
SIR and SKR and the associated projection systems, performance-oriented
optimal fault detection and faults estimation are achieved. To this end, a
design procedure is proposed that combines the control theoretic methods,
Hamiltonian systems with the associated Hamiltonian functions as well as
inner and co-inner systems, and the methods known in information geometry
like Bregman divergences and Legendre transform. The latter gives an
interpretation of the projection onto the system image manifold as a
geodesic projection in the context of a dually flat Riemannian manifold. In
addition, a fault estimation scheme is proposed on the basis of the concepts
of the uncertainty model and manifold of uncertain data.

\item Applications of the developed novel fault diagnosis schemes and
methods lead to numerous new results, including,

\begin{itemize}
\item determination of the Hamiltonian functions with respect to the nominal
dynamics and dynamics corrupted with uncertainties/faults, as given in
Theorem \ref{Theo3-2}, Corollary \ref{col3-1} as well as Theorem \ref%
{Theo4-1} and Corollary \ref{col4-1},

\item determination of Bregman divergences as evaluation functions,
corresponding to the projections onto the image manifold and the manifold of
uncertain data, proved in Theorem \ref{Theo3-3} and shown in (\ref{eq4-6}),
respectively,

\item proof of the interpretation that the projection onto the image
manifold is a geodesic projection, as given in Theorem \ref{Theo3-4},

\item proof of Theorem \ref{Theo3-5} that shows the evaluation function over
a time interval is a Bregman divergence as well,

\item proof of the interpretation of the projection onto the manifold of
uncertain data as a geodesic projection, given in Theorem \ref{Theo4-3}, and

\item demonstration of the projections onto the manifold of uncertain data
as a least square (LS) estimation subject to the uncertainty model in
Subsection \ref{subsec4-3-2}.
\end{itemize}
\end{itemize}

The draft is organised as follows. In Section 2, necessary preliminaries are
introduced. At first, the orthogonal projection-based fault detection scheme
for LTI systems, serving as background and inspiration, is briefly reviewed.
Further control-theoretic preliminaries include image and kernel
representations of nonlinear dynamic systems, Hamiltonian systems, inner and
co-inner systems. It is followed by an introduction to definitions and
computation of Bregman divergences and Legendre transform. They are used not
only as divergences for performance-oriented fault detection, but also for
inducing a dual Riemannian structure in the data space, which admits a
conjugate pair of affine connections. In this regard, the concepts of
divergences from a (data) vector to a manifold and geodesic projections are
described. At the end of this section, problems to be addressed in this
draft are specified. Section 3 is devoted to the projection and Bregman
divergence based fault detection. This work consists of three parts. The
first part deals with construction of a projection system by means of the
normalised SIR and its associated Hamiltonian systems, and the analysis of
the corresponding Hamiltonian functions under different operation
conditions. It is followed by the second part focusing on (i) computation of
the Bregman divergence from the data vector to its projection, which is
derived from the (dual) Hamiltonian function, (ii) its use for fault
detection, and (iii) its interpretation as a geodesic projection. The last
part in this section discusses about implementations issues, including
definition of a reliable evaluation function over an evaluation time
interval and determination of a threshold. In Section 4, two different but
strongly related topics are addressed. The first one is an SKR-based
projection, as a dual fault detection scheme to the SIR-based projection.
This work is analogue to the one presented in Section 3. Fault estimation is
the second topic. The focus of this study is on the application of an
SKR-based projection as a fault estimator and analysis of its estimation
performance. For both parts, the so-called uncertainty model and the
manifold of uncertain data serve as a common basis and are introduced at the
beginning of this section.

\bigskip

Throughout this paper, standard notations known in advanced control theory
and linear algebra are adopted. In addition, $\mathcal{L}_{2}$ is time
domain space of all square summable Lebesgue signals (signals with bounded
energy). For a transfer matrix $G(s),$ its conjugate is denoted by $G^{\sim
}(s)\left( =G^{T}(-s)\right) .$ In the context of Legendre transform, the
dual of vector $\alpha $ and function $\varphi $ is denoted by $\alpha
^{\times }$ and $\varphi ^{\times },$ respectively. Spaces and manifolds are
denoted by calligraphic letters such as $\mathcal{I},\mathcal{K}$ for image
and kernel subspaces (manifolds).

\section{Preliminaries}

\subsection{System factorisations and projection-based fault detection of
LTI systems\label{sec2-1}}

Consider an LTI system modelled by a transfer function matrix $G(s),$ whose
minimal state space representation is given by%
\begin{equation*}
G:%
\begin{cases}
\dot{x}(t)=Ax(t)+Bu(t) \\ 
y(t)=Cx(t)+Du(t),%
\end{cases}%
\end{equation*}%
where $x\in \mathbb{R}^{n},u\in \mathbb{R}^{p},y\in \mathbb{R}^{m}$ are the
state, input and output vectors, respectively, $A,B,C,D$ are system matrices
of appropriate dimensions. A right coprime factorisation (RCF) of $G$ is
given by $G(s)=N(s)M^{-1}(s)$ with the right coprime pair $\left( M,N\right) 
$ whose state space representations are given by%
\begin{equation}
M(s)=(A+BF,BV,F,V),N(s)=(A+BF,BV,C+DF,DV).
\end{equation}%
where $F$ is selected such that $A+BF$ is Hurwitz, $V$ as a pre-filter is an
invertible constant matrix. The RCF of $G$ can be interpreted as a state
feedback control system with $u=Fx+Vv,$ where $v$ serves as a reference
signal. An alternative model form of $y=Gu$ is the so-called stable image
representation (SIR) \cite{Ding2020,vanderSchaftbook} expressed in terms of $%
\left( M,N\right) $ by 
\begin{equation}
I_{G}(s):\left[ 
\begin{array}{c}
u(s) \\ 
y(s)%
\end{array}%
\right] =\left[ 
\begin{array}{c}
M(s) \\ 
N(s)%
\end{array}%
\right] v(s).
\end{equation}%
In this context, vector $v$ is understood as a latent variable.

\begin{Rem}
Hereafter, we may drop out the domain variable $s$\ or $t$ when there is no
risk of confusion.
\end{Rem}

The dual concepts to RCF and SIR are the so-called left coprime
factorisation (LCF) and stable kernel representation (SKR) of $G$\ {\cite%
{Vinnicombe-book}}. Denoted by 
\begin{equation*}
K_{G}(s)=\left[ 
\begin{array}{cc}
-\hat{N}(s) & \text{ }\hat{M}(s)%
\end{array}%
\right]
\end{equation*}%
with $\left( \hat{M},\hat{N}\right) $ as a left coprime pair, the state
space realisation of $K_{G}$ is 
\begin{equation*}
K_{G}:\left\{ 
\begin{array}{l}
\dot{\hat{x}}=\left( A-LC\right) \hat{x}+\left[ 
\begin{array}{cc}
B-LD & \text{ }-L%
\end{array}%
\right] \left[ 
\begin{array}{c}
u \\ 
y%
\end{array}%
\right] \\ 
r_{y}=-WC\hat{x}+\left[ 
\begin{array}{cc}
-WD & \text{ }W%
\end{array}%
\right] \left[ 
\begin{array}{c}
u \\ 
y%
\end{array}%
\right] ,%
\end{array}%
\right.
\end{equation*}%
which is an observer-based residual generator with $r_{y}$ as the residual
vector satisfying%
\begin{equation*}
r_{y}=Wr_{0},r_{0}=y-C\hat{x}-Du.
\end{equation*}%
As an observer-based residual generator, the observer gain $L$ is selected
such that $A-LC$ is Hurwitz, and $W$ is an invertible post-filter. It is
obvious that for $\hat{x}(0)=x(0),K_{G}$ satisfies 
\begin{equation}
r_{y}=K_{G}\left[ 
\begin{array}{c}
u \\ 
y%
\end{array}%
\right] =K_{G}I_{G}v=0\Longrightarrow \left[ 
\begin{array}{cc}
-\hat{N} & \text{ }\hat{M}%
\end{array}%
\right] \left[ 
\begin{array}{c}
M \\ 
N%
\end{array}%
\right] =0.  \label{eq2-31}
\end{equation}%
During normal process operations, the process data $\left( u,y\right) $
build a subspace in Hilbert space $\mathcal{L}_{2},$ called image subspace
of $G,$ which is explicitly defined by the SIR of $G$ and the latent
variable $v$ as follows 
\begin{equation}
\mathcal{I}_{G}=\left\{ \left[ 
\begin{array}{c}
u \\ 
y%
\end{array}%
\right] :\left[ 
\begin{array}{c}
u \\ 
y%
\end{array}%
\right] =\left[ 
\begin{array}{c}
M \\ 
N%
\end{array}%
\right] v,v\in \mathcal{L}_{2}\right\} .
\end{equation}%
By means of $K_{G}$, the kernel subspace of $G$ is defined as%
\begin{equation}
\mathcal{K}_{G}=\left\{ \left[ 
\begin{array}{c}
u \\ 
y%
\end{array}%
\right] :r_{y}=\left[ 
\begin{array}{cc}
-\hat{N} & \text{ }\hat{M}%
\end{array}%
\right] \left[ 
\begin{array}{c}
u \\ 
y%
\end{array}%
\right] =0,\left[ 
\begin{array}{c}
u \\ 
y%
\end{array}%
\right] \in \mathcal{L}_{2}\right\} ,
\end{equation}%
which is, due to relation (\ref{eq2-31}), identical with $\mathcal{I}_{G},$
i.e. $\mathcal{K}_{G}=\mathcal{I}_{G}$ {\cite{Vinnicombe-book}}. The
complementary subspace of $\mathcal{I}_{G}$ is denoted by $\mathcal{I}%
_{G}^{\bot }$.

\bigskip

Let 
\begin{equation*}
I_{G,0}(s)=\left[ 
\begin{array}{c}
M_{0}(s) \\ 
N_{0}(s)%
\end{array}%
\right] ,K_{G,0}(s)=\left[ 
\begin{array}{cc}
-\hat{N}_{0}(s) & \text{ }\hat{M}_{0}(s)%
\end{array}%
\right]
\end{equation*}%
denote the normalised SIR and SKR of $G,$ defined by 
\begin{equation}
I_{G,0}^{\sim }I_{G,0}=N_{0}^{\sim }N_{0}+M_{0}^{\sim
}M_{0}=I,K_{G,0}K_{G,0}^{\sim }=\hat{N}_{0}\hat{N}_{0}^{\sim }+\hat{M}_{0}%
\hat{M}_{0}^{\sim }=I.  \label{eq2-32}
\end{equation}%
On account of (\ref{eq2-31}) and (\ref{eq2-32}), we have%
\begin{equation}
\left[ 
\begin{array}{c}
I_{G,0}^{\sim } \\ 
K_{G,0}%
\end{array}%
\right] \left[ 
\begin{array}{cc}
I_{G,0} & \text{ }K_{G,0}^{\sim }%
\end{array}%
\right] =\left[ 
\begin{array}{cc}
M_{0}^{\sim } & \text{ }N_{0}^{\sim } \\ 
-\hat{N}_{0} & \text{ }\hat{M}_{0}%
\end{array}%
\right] \left[ 
\begin{array}{cc}
M_{0} & \text{ }-\hat{N}_{0}^{\sim } \\ 
N_{0} & \text{ }\hat{M}_{0}^{\sim }%
\end{array}%
\right] =\left[ 
\begin{array}{cc}
I & \text{ }0 \\ 
0 & \text{ }I%
\end{array}%
\right] .  \label{eq2-20}
\end{equation}%
The reader is referred to \cite{Hoffmann1996} for more details about the
normalised SIR and SKR and their computation. It is noteworthy that the
normalised SIR and SKR of $G$ are so-called inner and co-inner systems,
respectively \cite{Hoffmann1996}.

\bigskip

Now, we are in a position to introduce the basic idea and algorithms of
projection-based optimal fault detection system. Let $\mathcal{P}_{\mathcal{V%
}}$ be an operator defined on a subspace $\mathcal{V}$ in Hilbert space that
is endowed with the inner product, 
\begin{equation*}
\left\langle \alpha ,\beta \right\rangle =\int\limits_{-\infty }^{\infty
}\alpha ^{T}(t)\beta (t)dt,\alpha ,\beta \in \mathcal{V\subset L}_{2}.
\end{equation*}%
If $\mathcal{P}_{\mathcal{V}}$ is idempotent and self-adjoint, namely 
\begin{equation}
\forall \alpha ,\beta \in \mathcal{V},\mathcal{P}_{\mathcal{V}}^{2}=\mathcal{%
P}_{\mathcal{V}},\left\langle \mathcal{P}_{\mathcal{V}}\alpha ,\beta
\right\rangle =\left\langle \alpha ,\mathcal{P}_{\mathcal{V}}\beta
\right\rangle ,
\end{equation}%
it is an operator of an orthogonal projection onto $\mathcal{V}$ \cite%
{Kato_book}. If the subspace $\mathcal{V}$ is closed, the distance between $%
\beta $ and $\mathcal{V},dist\left( \beta ,\mathcal{V}\right) ,$ is defined
as%
\begin{equation}
dist\left( \beta ,\mathcal{V}\right) =\inf_{\alpha \in \mathcal{V}%
}\left\Vert \beta -\alpha \right\Vert _{2},  \label{eq2-0}
\end{equation}%
which can be computed as 
\begin{equation*}
dist\left( \beta ,\mathcal{V}\right) =\left\Vert \beta -\mathcal{P}_{%
\mathcal{V}}\beta \right\Vert _{2}=\left\Vert \mathcal{P}_{\mathcal{V}^{\bot
}}\beta \right\Vert _{2}.
\end{equation*}%
It is well-known that the image subspace $\mathcal{I}_{G}$ is closed in $%
\mathcal{L}_{2}$ and $\mathcal{P}_{\mathcal{I}_{G}}:\mathcal{L}%
_{2}\rightarrow \mathcal{L}_{2},$ 
\begin{equation*}
\mathcal{P}_{\mathcal{I}_{G}}\left( \left[ 
\begin{array}{c}
u \\ 
y%
\end{array}%
\right] \right) :=I_{G,0}I_{G,0}^{\sim }\left[ 
\begin{array}{c}
u \\ 
y%
\end{array}%
\right]
\end{equation*}%
defines an orthogonal projection onto $\mathcal{I}_{G}$ \cite%
{Vinnicombe-book} with $\mathcal{P}_{\mathcal{I}_{G}}$ denoting the
projection operator. Note that operator $\mathcal{I-P}_{\mathcal{I}_{G}}:%
\mathcal{L}_{2}\rightarrow \mathcal{L}_{2},$ 
\begin{equation}
\left( \mathcal{I-P}_{\mathcal{I}_{G}}\right) \left[ 
\begin{array}{c}
u \\ 
y%
\end{array}%
\right] =\left( I-I_{G,0}I_{G,0}^{\sim }\right) \left[ 
\begin{array}{c}
u \\ 
y%
\end{array}%
\right] ,  \label{eq2-21b}
\end{equation}%
defines an orthogonal projection onto the orthogonal complement of $\mathcal{%
I}_{G}$. Consequently, any process data can be written as 
\begin{equation*}
\left[ 
\begin{array}{c}
u \\ 
y%
\end{array}%
\right] =\mathcal{P}_{\mathcal{I}_{G}}\left[ 
\begin{array}{c}
u \\ 
y%
\end{array}%
\right] +\mathcal{P}_{\mathcal{I}_{G}^{\bot }}\left[ 
\begin{array}{c}
u \\ 
y%
\end{array}%
\right] ,\mathcal{P}_{\mathcal{I}_{G}^{\bot }}=\mathcal{I-P}_{\mathcal{I}%
_{G}}.
\end{equation*}%
In the context of one-class classification, faulty operations are detected
if 
\begin{equation*}
\left[ 
\begin{array}{c}
u \\ 
y%
\end{array}%
\right] \notin \mathcal{I}_{G}\Longrightarrow \mathcal{P}_{\mathcal{I}%
_{G}^{\bot }}\left[ 
\begin{array}{c}
u \\ 
y%
\end{array}%
\right] \neq 0,
\end{equation*}%
and $\mathcal{P}_{\mathcal{I}_{G}^{\bot }}\left[ 
\begin{array}{c}
u \\ 
y%
\end{array}%
\right] $ is sufficiently large (with respect to a defined threshold). From
the implementation point of view, it is of interest to notice that, due to (%
\ref{eq2-20}), it holds%
\begin{eqnarray}
\left[ 
\begin{array}{cc}
I_{G,0} & \text{ }K_{G,0}^{\sim }%
\end{array}%
\right] \left[ 
\begin{array}{c}
I_{G,0}^{\sim } \\ 
K_{G,0}%
\end{array}%
\right] &=&I\Longleftrightarrow I-I_{G,0}I_{G,0}^{\sim }=K_{G,0}^{\sim
}K_{G,0}\Longrightarrow  \notag \\
\left\Vert \left( \mathcal{I-P}_{\mathcal{I}_{G}}\right) \left( \left[ 
\begin{array}{c}
u \\ 
y%
\end{array}%
\right] \right) \right\Vert _{2} &=&\left\Vert K_{G,0}^{\sim }K_{G,0}\left[ 
\begin{array}{c}
u \\ 
y%
\end{array}%
\right] \right\Vert _{2}=\left\Vert K_{G,0}\left[ 
\begin{array}{c}
u \\ 
y%
\end{array}%
\right] \right\Vert _{2}=\left\Vert r_{y}\right\Vert _{2}.  \label{eq2-21}
\end{eqnarray}%
Consequently, when the $\mathcal{L}_{2}$-norm is used for residual
evaluation, the projection $\mathcal{P}_{\mathcal{I}_{G}}$ and the
corresponding residual \ generation $\left( \mathcal{I-P}_{\mathcal{I}%
_{G}}\right) \left( \left[ 
\begin{array}{c}
u \\ 
y%
\end{array}%
\right] \right) $ can be equivalently realised in the observer-based fault
detection framework \cite{Ding2008}.

\bigskip

In a nutshell, an orthogonal projection onto a subspace in Hilbert space
enables us to compute the distance between a collected data and the nominal
system operation presented by the image/kernel subspace of the system under
consideration. Thanks to the orthogonality, it holds%
\begin{align}
\left[ 
\begin{array}{c}
u \\ 
y%
\end{array}%
\right] & =\mathcal{P}_{\mathcal{I}_{G}}\left( \left[ 
\begin{array}{c}
u \\ 
y%
\end{array}%
\right] \right) +\mathcal{P}_{\mathcal{I}_{G}^{\bot }}\left[ 
\begin{array}{c}
u \\ 
y%
\end{array}%
\right] \Longrightarrow  \notag \\
\left\Vert \left[ 
\begin{array}{c}
u \\ 
y%
\end{array}%
\right] \right\Vert _{2}^{2}& =\left\Vert \mathcal{P}_{\mathcal{I}%
_{G}}\left( \left[ 
\begin{array}{c}
u \\ 
y%
\end{array}%
\right] \right) \right\Vert _{2}^{2}\mathcal{+}\left\Vert \mathcal{P}_{%
\mathcal{I}_{G}}\left( \left[ 
\begin{array}{c}
u \\ 
y%
\end{array}%
\right] \right) \right\Vert _{2}^{2}\Longrightarrow  \notag \\
\left\Vert r_{y}\right\Vert _{2}^{2}& =\left\Vert K_{G,0}\left[ 
\begin{array}{c}
u \\ 
y%
\end{array}%
\right] \right\Vert _{2}^{2}=\left\Vert \left[ 
\begin{array}{c}
u \\ 
y%
\end{array}%
\right] \right\Vert _{2}^{2}-\left\Vert \mathcal{P}_{\mathcal{I}_{G}}\left( %
\left[ 
\begin{array}{c}
u \\ 
y%
\end{array}%
\right] \right) \right\Vert _{2}^{2},  \label{eq2-21a}
\end{align}%
which results in an optimal fault detection in LTI systems. We would like to
emphasise that this scheme is a one-class fault detection, which is achieved
solely based on the nominal system dynamics. An extension to two-class fault
detection as well as to fault isolation (multi-class classification) is
straightforward \cite{ding2022}.

\subsection{Kernel and image representations of nonlinear systems, and
manifolds}

A nonlinear dynamic system $\Sigma $ can be generally modelled as%
\begin{equation*}
\Sigma :\left\{ 
\begin{array}{l}
\dot{x}(t)=f(x(t),u(t)) \\ 
y(t)=h(x(t),u(t)),%
\end{array}%
\right.
\end{equation*}%
where $x\in \mathbb{R}^{n}$ denotes the state vector, $u\in \mathbb{R}^{p}$
and $y\in \mathbb{R}^{m}$ represent system input and output vectors,
respectively. $f(x,u)$ and $h(x,u)$ are continuously differentiable
nonlinear functions with appropriate dimensions. In our work, for the sake
of simplicity, the affine form of $\Sigma ,$ 
\begin{equation}
\Sigma :\left\{ 
\begin{array}{l}
\dot{x}=a(x)+B(x)u,x(0)=x_{0} \\ 
y=c(x)+D(x)u,%
\end{array}%
\right.  \label{eq2-1}
\end{equation}%
is under consideration. Here, $x_{0}$ represents the initial value of $x(t),$
and $a(x),B(x),c(x)$ and $D(x)$ are continuously differentiable and of
appropriate dimensions. The affine system (\ref{eq2-1}) is a class of
nonlinear systems which are widely adopted in nonlinear system research and
can be considered as a natural extension of linear systems.

\bigskip

Roughly speaking, given system $\Sigma ,$ a stable system $\Sigma _{\mathcal{%
I}}$ is a stable image representation (SIR) of $\Sigma $ if for all $%
\mathcal{L}_{2}$-bounded input $u$ and its $\mathcal{L}_{2}$-bounded
response $y=\Sigma (u)$, there exists an $\mathcal{L}_{2}$-bounded $v\in 
\mathbb{R}^{p}$ such that \ \cite{vanderSchaftbook} 
\begin{equation*}
\left[ 
\begin{array}{c}
u \\ 
y%
\end{array}%
\right] =\Sigma _{\mathcal{I}}\left( v\right) .
\end{equation*}%
The state space representation of the SIR $\Sigma _{\mathcal{I}}\left(
v\right) $ is an essential model form used in our subsequent work and is
given by 
\begin{gather}
\Sigma _{\mathcal{I}}:\hspace{-2pt}%
\begin{cases}
\dot{x}=a_{I}(x)+B_{I}(x)v \\ 
\left[ 
\begin{array}{c}
u \\ 
y%
\end{array}%
\right] =c_{I}(x)+D_{I}(x)v,%
\end{cases}
\label{eq2-2} \\
a_{I}(x)=a(x)+B(x)g(x),B_{I}(x)=B(x)V(x),  \notag \\
c_{I}(x)=\left[ 
\begin{array}{c}
g(x) \\ 
c(x)+D(x)g(x)%
\end{array}%
\right] ,D_{I}(x)=\left[ 
\begin{array}{c}
V(x) \\ 
D(x)V(x)%
\end{array}%
\right] ,  \notag
\end{gather}%
where $g(x)$ is designed such that $a_{I}(x)$ is stable, and $V(x)\in 
\mathbb{R}^{p\times p}$ is invertible. In the context of feedback control
systems, 
\begin{equation*}
u=g(x)+V(x)v
\end{equation*}%
is understood as a controller with a state feedback $g(x),$ feed-forward
controller $V(x)$ and ${v}$ as the reference signal. Note that 
\begin{equation*}
D_{I}^{T}(x)D_{I}(x)=V^{T}(x)\left( I+D^{T}(x)D(x)\right) V(x)
\end{equation*}%
is invertible.

\bigskip

The so-called stable kernel representation (SKR) of $\Sigma $ is another
system model form and useful\ in our subsequent study. The SKR of $\Sigma ,$
denoted by $\Sigma _{\mathcal{K}},$ is a stable system that satisfies,
roughly speaking, 
\begin{equation}
r_{y}=\Sigma _{\mathcal{K}}\left( 
\begin{array}{c}
u \\ 
y%
\end{array}%
\right) =0,r_{y}\in \mathbb{R}^{m},  \label{eq2-2b}
\end{equation}%
when there exists no uncertainty in the system. The state space
representation of SKR $\Sigma _{\mathcal{K}}$ is described by 
\begin{align}
\Sigma _{\mathcal{K}}& :\left\{ 
\begin{array}{l}
\dot{\hat{x}}=a(\hat{x})+B(\hat{x})u+L(\hat{x})\left( y-\hat{y}\right) ,\hat{%
x}(0)=\hat{x}_{0} \\ 
\text{ \ }=a_{K}(\hat{x})+B_{K}(\hat{x})\left[ 
\begin{array}{c}
u \\ 
y%
\end{array}%
\right] \\ 
r_{y}=W(\hat{x})r_{0}=c_{K}(\hat{x})+D_{K}(\hat{x})\left[ 
\begin{array}{c}
u \\ 
y%
\end{array}%
\right] \\ 
r_{0}=y-\hat{y},\hat{y}=c\left( \hat{x}\right) +D\left( \hat{x}\right) u,%
\end{array}%
\right.  \label{eq2-3} \\
a_{K}(\hat{x})& =a(\hat{x})-L(\hat{x})c\left( \hat{x}\right) ,B_{K}(\hat{x})=%
\left[ 
\begin{array}{cc}
B(\hat{x})-L(\hat{x})D\left( \hat{x}\right) & \text{ }L(\hat{x})%
\end{array}%
\right] ,  \notag \\
c_{K}(\hat{x})& =-W(\hat{x})c\left( \hat{x}\right) ,D_{K}(\hat{x})=\left[ 
\begin{array}{cc}
-W(\hat{x})D\left( \hat{x}\right) & \text{ }W(\hat{x})%
\end{array}%
\right] ,  \notag
\end{align}%
where $L(\hat{x})$ is designed such that $a_{K}(\hat{x})$ is stable and $W(%
\hat{x})\in \mathbb{R}^{m\times m}$ is invertible. In the detection context, 
$r_{0}$ builds a residual vector, $W$ is a post-filter, and it holds 
\begin{equation*}
D_{K}^{T}(\hat{x})D_{K}(\hat{x})=W(\hat{x})\left( I+D(\hat{x})D^{T}(\hat{x}%
)\right) W^{T}(\hat{x})
\end{equation*}%
being invertible. Note that for $\hat{x}(0)=\hat{x}%
_{0}=x(0)=x_{0},r_{y}=r_{0}=0.$ For this case, we denote $\Sigma _{\mathcal{K%
}}$ by $\Sigma _{\mathcal{K}}^{x_{0}}.$ It is apparent that $\forall {v,}$%
\begin{equation}
\Sigma _{\mathcal{K}}^{x_{0}}\circ \Sigma _{\mathcal{I}}\left( v\right)
=\Sigma _{\mathcal{K}}^{x_{0}}\left( 
\begin{array}{c}
u \\ 
y%
\end{array}%
\right) =0\Longrightarrow \Sigma _{\mathcal{K}}^{x_{0}}\circ \Sigma _{%
\mathcal{I}}=0.  \label{eq2-4}
\end{equation}%
In terms of the concepts of SIR and SKR, we are in a position to introduce
the definitions of image and kernel manifolds. They are lower-dimensional
manifolds embedded in the data space of $\left( u,y\right) $. This manner of
modelling dynamic systems enables us to address fault
detection/classification issues directly in the data space, rather than on
the basis of the system input-output relations.

\begin{Def}
Given system (\ref{eq2-1}), its SIR (\ref{eq2-2}) and SKR (\ref{eq2-3}),%
\begin{align}
\mathcal{I}_{\Sigma }& =\left\{ \left[ 
\begin{array}{c}
u \\ 
y%
\end{array}%
\right] :\left[ 
\begin{array}{c}
u \\ 
y%
\end{array}%
\right] =\Sigma _{\mathcal{I}}\left( v\right) ,v\in \mathcal{L}_{2}\right\} ,
\\
\mathcal{K}_{\Sigma }& =\left\{ \left[ 
\begin{array}{c}
u \\ 
y%
\end{array}%
\right] \in \mathcal{L}_{2}:\Sigma _{\mathcal{K}}^{x_{0}}\left( \left[ 
\begin{array}{c}
u \\ 
y%
\end{array}%
\right] \right) =0\right\}
\end{align}%
are called image and kernel sub-manifold of $\Sigma ,$ respectively.
\end{Def}

Note that, due to relation (\ref{eq2-4}), it holds $\mathcal{I}_{\Sigma }=%
\mathcal{K}_{\Sigma }.$

\subsection{Hamiltonian, inner systems, and normalised SIR and SKR}

Recall that for LTI systems, the normalised SIR and SKR are defined in terms
of the transfer functions of the RCF and LCF as well as their conjugates.
For nonlinear systems, the well-established Hamiltonian system technique can
be applied for this purpose. In a broad sense, given a nonlinear system as
an operator, the associated \ Hamiltonian system can be interpreted as a
certain configuration of the operator and its adjoint. In this regard, we
first introduce some essential concepts and definitions known in the
literature.

\subsubsection{Hamiltonian systems}

Consider the affine system (\ref{eq2-1}) and assume that it is stable. The
Hamiltonian extension of $\Sigma ,$ introduced \ in \cite{Schaft-book1987},
is a dynamic system described by 
\begin{equation}
\begin{cases}
\dot{x}=a(x)+B(x)u \\ 
\dot{\lambda}=-\left( \frac{\partial a(x)}{\partial x}+\frac{\partial B(x)}{%
\partial x}u\right) ^{T}\lambda -\left( \frac{\partial c\left( x\right) }{%
\partial x}+\frac{\partial D(x)}{\partial x}u\right) ^{T}u_{a} \\ 
y=c(x)+D(x)u \\ 
y_{a}=B^{T}(x)\lambda +D^{T}(x)u_{a},y_{a}\in \mathbb{R}^{p},u_{a}\in 
\mathbb{\ R}^{m},%
\end{cases}
\label{eq2-5}
\end{equation}%
with state variables $\left( x,\lambda \right) ,$ input variables $\left(
u,u_{a}\right) $ and output variables $\left( y,y_{a}\right) $. $\lambda \in 
\mathbb{\ R}^{n}$ is also called co-state variable. Connecting $y$ and $%
u_{a},u_{a}=y=c(x)+D(x)u,$ and defining a Hamiltonian function 
\begin{equation*}
H\left( x,\lambda ,u\right) =\frac{1}{2}\left( c(x)+D(x)u\right) ^{T}\left(
c(x)+D(x)u\right) +\lambda ^{T}\left( a(x)+B(x)u\right)
\end{equation*}%
lead to the Hamiltonian system $y_{a}=\left( D\Sigma \right) ^{T}\circ
\Sigma \left( u\right) ,$ whose state space representation can be written in
a compact form%
\begin{equation}
\left( D\Sigma \right) ^{T}\circ \Sigma :%
\begin{cases}
\dot{x}=\frac{\partial H}{\partial \lambda }\left( x,\lambda ,u\right) \\ 
\dot{\lambda}=-\frac{\partial H}{\partial x}\left( x,\lambda ,u\right) \\ 
y_{a}=\frac{\partial H}{\partial u}\left( x,\lambda ,u\right)%
\end{cases}
\label{eq2-6}
\end{equation}%
with input vector $u$ and output vector $y_{a}.$ Here, $D\Sigma $ denotes
the Frechet derivative of $\Sigma $ \cite{Schaft-book1987}. Differently,
connecting $u$ and $y_{a}$ of the Hamiltonian extension (\ref{eq2-5}),
namely $u=y_{a}=B^{T}(x)\lambda +D^{T}(x)u_{a},$ results in%
\begin{gather}
\Sigma \circ \left( D\Sigma \right) ^{T}:%
\begin{cases}
\dot{x}=a(x)+B(x)\left( B^{T}(x)\lambda +D^{T}(x)u_{a}\right) \\ 
\dot{\lambda}=-\left( \frac{\partial a(x)}{\partial x}+\frac{\partial \bar{B}%
(x,\lambda ,u_{a})}{\partial x}\right) ^{T}\lambda -\left( \frac{\partial
c\left( x\right) }{\partial x}+\frac{\partial \bar{D}(x,\lambda ,u_{a})}{%
\partial x}\right) ^{T}u_{a} \\ 
y=c(x)+D(x)B^{T}(x)\lambda +D(x)D^{T}(x)u_{a},%
\end{cases}
\label{eq2-7} \\
\bar{B}(x,\lambda ,u_{a})=B(x)\left( B^{T}(x)\lambda +D^{T}(x)u_{a}\right) ,
\notag \\
\bar{D}(x,\lambda ,u_{a})=D(x)B^{T}(x)\lambda +D(x)D^{T}(x)u_{a}.  \notag
\end{gather}%
Let 
\begin{equation}
H\left( x,\lambda ,u_{a}\right) =\frac{1}{2}%
u_{a}^{T}D(x)D^{T}(x)u_{a}+c^{T}(x)u_{a}+\lambda ^{T}\left( a(x)+\frac{1}{2}%
B(x)B^{T}(x)\lambda +B(x)D^{T}(x)u_{a}\right)  \label{eq2-8}
\end{equation}%
be a Hamiltonian function. Accordingly, $\Sigma \circ \left( D\Sigma \right)
^{T}$ can be compactly written as%
\begin{equation}
\Sigma \circ \left( D\Sigma \right) ^{T}:%
\begin{cases}
\dot{x}=\frac{\partial H}{\partial \lambda }\left( x,\lambda ,u_{a}\right)
\\ 
\dot{\lambda}=-\frac{\partial H}{\partial x}\left( x,\lambda ,u_{a}\right)
\\ 
y=\frac{\partial H}{\partial u_{a}}\left( x,\lambda ,u_{a}\right) .%
\end{cases}
\label{eq2-9}
\end{equation}

\subsubsection{Inner and co-inner, normalised SIR and SKR}

As mentioned in Subsection 2.1, for LTI systems, the normalised SIR and SKR
are inner and co-inner, respectively. This motivates us to review the
concepts and existence conditions of nonlinear inner and co-inner systems,
before addressing the normalised SIR and SKR issues. The presented results
can be found e.g. in \cite{Scherpen1994,Ball1996,vanderSchaftbook}.

\begin{Def}
\label{Def2-2}Consider the stable nonlinear affine system (\ref{eq2-1}). $%
\Sigma $ is inner if the Hamiltonian \ system (\ref{eq2-6}) satisfies%
\begin{equation}
y_{a}=u,  \label{eq2-10}
\end{equation}%
and lossless, and it is co-inner, when it holds, for the Hamiltonian system (%
\ref{eq2-9}), 
\begin{equation}
y=u_{a}  \label{eq2-11}
\end{equation}%
and lossless. Here, $\Sigma $ is called lossless, if there exists a storage
function $V(x)\geq 0,V(0)=0$ so that 
\begin{equation}
V\left( x\left( t_{2}\right) \right) -V\left( x\left( t_{1}\right) \right) =%
\frac{1}{2}\int_{t_{1}}^{t_{2}}\left( u^{T}u-y^{T}y\right) d\tau .
\label{eq2-11a}
\end{equation}
\end{Def}

\begin{Le}
\label{Lemma2-1}The nonlinear affine system (\ref{eq2-1}) is inner, if there
exists $V(x)\geq 0$ such that the following equations are feasible 
\begin{align}
& V_{x}^{T}\left( x\right) a(x)+\frac{1}{2}c^{T}(x)c(x)=0, \\
& V_{x}^{T}\left( x\right) B(x)+c^{T}(x)D(x)=0, \\
& D^{T}(x)D(x)=I,
\end{align}%
where $V_{x}(x)=\frac{\partial V(x)}{\partial x}$. And, it is co-inner, if
there exists $V(x)\geq 0$ such that the following equations are feasible 
\begin{align}
& V_{x}^{T}\left( x\right) a(x)+\frac{1}{2}V_{x}^{T}\left( x\right)
B(x)B^{T}(x)V_{x}\left( x\right) =0, \\
& c^{T}(x)+V_{x}^{T}\left( x\right) B(x)D^{T}(x)=0, \\
& D(x)D^{T}(x)=I.
\end{align}
\end{Le}

\begin{Def}
The SIR $\Sigma _{\mathcal{I}}$ and SKR $\Sigma _{\mathcal{K}}$ given in (%
\ref{eq2-2}) and (\ref{eq2-3}), respectively, are normalised, when $\left(
g(x),V(x)\right) $ and $\left( L(x),W(x)\right) $ are set such that $\Sigma
_{\mathcal{I}}$ is inner and $\Sigma _{\mathcal{K}}$ is co-inner.
\end{Def}

\begin{Theo}
\label{Theo2-1}The SIR (\ref{eq2-2}) is normalised, if there exists $%
P(x)\geq 0$ such that the following equation is feasible%
\begin{equation}
P_{x}^{T}(x)a(x)+\frac{1}{2}c^{T}\left( x\right) c(x)-\frac{1}{2}%
g^{T}(x)\left( I+D^{T}(x)D(x)\right) g(x)=0,  \label{eq2-15a}
\end{equation}%
and $\left( g(x),V(x)\right) $ are set to be 
\begin{align}
g(x)& =-\left( I+D^{T}(x)D(x)\right) ^{-1}\left(
B^{T}(x)P_{x}(x)+D^{T}(x)c(x)\right) , \\
V(x)& =V^{T}(x)=\left( I+D^{T}(x)D(x)\right) ^{-1/2}.  \label{eq2-15b}
\end{align}%
And the SKR (\ref{eq2-3}) is normalised, if there exists $V(\hat{x})\geq 0$
solving%
\begin{gather}
V_{\hat{x}}^{T}\left( \hat{x}\right) a(\hat{x})+\frac{1}{2}V_{\hat{x}%
}^{T}\left( \hat{x}\right) B(\hat{x})B^{T}(\hat{x})V_{\hat{x}}\left( \hat{x}%
\right)  \notag \\
-\frac{1}{2}\left( c^{T}(\hat{x})+V_{\hat{x}}^{T}\left( \hat{x}\right) B(%
\hat{x})D^{T}(\hat{x})\right) \left( I+D(\hat{x})D^{T}(\hat{x})\right)
^{-1}\left( c^{T}(\hat{x})+V_{\hat{x}}^{T}\left( \hat{x}\right) B(\hat{x}%
)D^{T}(\hat{x})\right) ^{T}=0,  \label{eq2-15c}
\end{gather}%
and $\left( L(x),W(x)\right) $ are set to be%
\begin{align}
V_{\hat{x}}^{T}\left( \hat{x}\right) L(\hat{x})& =\left( c^{T}(\hat{x})+V_{%
\hat{x}}^{T}\left( \hat{x}\right) B(\hat{x})D^{T}(\hat{x})\right) \left( I+D(%
\hat{x})D^{T}(\hat{x})\right) ^{-1},  \label{eq2-15d} \\
W(\hat{x})& =\left( I+D(\hat{x})D^{T}(\hat{x})\right) ^{-1/2}.
\end{align}
\end{Theo}

The proof of this theorem follows immediately from Lemma \ref{Lemma2-1} and
thus is omitted.

\begin{Rem}
The normalised SIR can be interpreted as an optimal controller under a
quadric cost function associated to the Hamiltonian functions \cite%
{vanderSchaftbook}. A normalised SKR based estimator results in a LS
estimation of uncertainties corrupted in data $\left( u,y\right) ,$ as will
be demonstrated in Subsection 4.3.
\end{Rem}

\subsection{Bregman divergence, Legendre transform, and geodesic projection 
\label{subsec2-4}}

For our purpose, we restrict our attention to the following definition of a
Bregman divergence. Let $\varphi :\mathbb{R}^{n}\rightarrow \mathbb{R}$ be\
a continuously-differentiable, strictly convex function. Bregman divergence
from $\alpha \in \mathbb{R}^{n}$ to $\beta \in \mathbb{R}^{n},$ denoted by $%
D_{\varphi }\left[ \alpha :\beta \right] ,$ is defined by 
\begin{equation*}
D_{\varphi }\left[ \alpha :\beta \right] =\varphi \left( \alpha \right)
-\varphi \left( \beta \right) -\left( \frac{\partial \varphi \left( \beta
\right) }{\partial \beta }\right) ^{T}\left( \alpha -\beta \right) .
\end{equation*}%
The function $\varphi $ is also called generating function. Bregman
divergences were introduced by Bregman as a measure of difference between
two points in a space \cite{BREGMAN1967}. A Bregman divergence is of all
properties of a divergence. In our subsequent work, the dual divergence to $%
D_{\varphi }\left[ \alpha :\beta \right] $ plays an important role. To
define it, we first consider a Legendre transform. Let 
\begin{equation*}
\alpha ^{\times }=\frac{\partial }{\partial \alpha }\varphi \left( \alpha
\right) \in \mathbb{R}^{n}.
\end{equation*}%
The function%
\begin{equation*}
\varphi ^{\times }\left( \alpha ^{\times }\right) :=\alpha ^{T}\alpha
^{\times }-\varphi \left( \alpha \right)
\end{equation*}%
is called the Legendre dual of $\varphi .$ It is known \cite{Amari2016} that%
\begin{equation*}
\alpha ^{\times }=\frac{\partial }{\partial \alpha }\varphi \left( \alpha
\right) ,\alpha =\frac{\partial }{\partial \alpha ^{\times }}\varphi
^{\times }\left( \alpha ^{\times }\right) ,
\end{equation*}%
form a dualistic structure, and $\varphi ^{\times }\left( \alpha ^{\times
}\right) $ is convex with respect to $\alpha ^{\times }.$

\bigskip

By means of the Legendre dual of $\varphi ,$ the dual divergence to $%
D_{\varphi }\left[ \alpha :\beta \right] $ is defined by 
\begin{equation*}
D_{\varphi ^{\times }}\left[ \alpha ^{\times }:\beta ^{\times }\right]
=\varphi \left( \alpha ^{\times }\right) -\varphi \left( \beta ^{\times
}\right) -\left( \frac{\partial \varphi ^{\times }\left( \beta ^{\times
}\right) }{\partial \beta ^{\times }}\right) ^{T}\left( \alpha ^{\times
}-\beta ^{\times }\right) .
\end{equation*}%
By some straightforward computations \cite{Amari2016}, we obtain 
\begin{align}
D_{\varphi }\left[ \alpha :\beta \right] & =\varphi \left( \alpha \right)
+\varphi ^{\times }\left( \beta ^{\times }\right) -\alpha ^{T}\beta ^{\times
},  \label{eq2-22} \\
D_{\varphi ^{\times }}\left[ \alpha ^{\times }:\beta ^{\times }\right] &
=D_{\varphi }\left[ \beta :\alpha \right] .  \label{eq2-22a}
\end{align}%
The introduction of the dual divergence is not only useful from the
computational point of view. It is the theoretical basis for the definition
and computation of the so-called geodesic projection onto a manifold.
Informally speaking, a geodesic projection can be understood as an extension
of the concept of an orthogonal projection onto a vector subspace. Below, we
shortly highlight the basic ideas and some results on the Bregman
divergence-based geodesic project described in \cite{Amari2016,Nielsen2020},
and restrict the introduction only to the so-called projection theorem that
is needed for our fault detection study. The reader is referred to \cite%
{Amari2016,Nielsen2020} and references therein for more details.

\bigskip

In his monograph on Information geometry \cite{Amari2016}, Amari begins the
introduction to a dually flat Riemannian manifold with the following
statement, ``\textit{we can establish a new edifice of differential
geometry, by treating a pair of affine connections which are dually coupled
with respect to the Riemannian metric. Such a structure ... is the heart of
information geometry. ... As an important special case, we study a dually
flat Riemannian manifold. It may be regarded as a dualistic extension of the
Euclidean space.}" Given two convex functions connected by the Legendre
transform $\alpha ^{\times }=$ $\nabla F\left( \alpha \right) ,\alpha
=\nabla F^{\times }\left( \alpha ^{\times }\right) $, Bregman and dual
Bregman divergences, $D_{\varphi }\left[ \alpha :\beta \right] =:$ $F\left(
\alpha \right) ,D_{\varphi ^{\times }}\left[ \alpha ^{\times }:\beta
^{\times }\right] =:F^{\times }\left( \alpha ^{\times }\right) $ on a
manifold $\mathcal{M}$ are defined. On this basis, metric tensors 
\begin{equation*}
G:=\nabla ^{2}F\left( \alpha \right) ,G^{\times }:=\nabla ^{2}F^{\times
}\left( \alpha ^{\times }\right)
\end{equation*}%
provide $\mathcal{M}$ with a dual Riemannian structure that admits a
conjugate pair of affine connections. In \cite{Amari2016}, it is proved that
these two affine connections are curvature- and torsion-free, and thus the
corresponding Riemannian structure is called dually flat. Consequently, any
pair of points in $\mathcal{M}$ can be linked by a straight line either in
the \ coordinate system $\alpha $ or in its dual coordinate system $\alpha
^{\times }$ \cite{Amari2016,Nielsen2020}. A main result related to a dually
flat structured manifold is the proof of the Generalised Pythagorean Theorem
by Nagaoka and Amari \cite{Amari1982}. Based on it, the projection theorem
given below was proved \cite{Amari2016}.

\bigskip

Let $\mathcal{M}\subset \mathbb{R}^{n}$ denote a dually flat manifold.

\begin{Def}
Given $\alpha \in \mathcal{M}$ and a submanifold $\mathcal{S}\subset 
\mathcal{M},$ the divergence from $\alpha $ to $\mathcal{S}$ is defined by%
\begin{equation*}
D_{\varphi }\left[ \alpha :\mathcal{S}\right] =\min_{\beta \in \mathcal{S}%
}D_{\varphi }\left[ \alpha :\beta \right] .
\end{equation*}
\end{Def}

A curve is said to be orthogonal to $\mathcal{S}$ when its tangent vector is
orthogonal to any tangent vectors of $\mathcal{S}$ at the intersection.

\begin{Def}
$\hat{\alpha}_{\mathcal{S}}\in \mathcal{S}$ is called geodesic projection of 
$\alpha $ to $\mathcal{S}$ when the geodesic connecting $\alpha $ and $\hat{%
\alpha}_{\mathcal{S}}$ is orthogonal to $\mathcal{S}$. Dually, $\hat{\alpha}%
_{\mathcal{S}}^{\times }\in \mathcal{S}$ is the dual geodesic projection of $%
\alpha $ to $\mathcal{S}$, when the dual geodesic connecting $\alpha $ and $%
\hat{\alpha}_{\mathcal{S}}^{\times }$ is orthogonal to $\mathcal{S}.$
\end{Def}

\begin{Theo}
\label{Theo2-2}(Projection theorem, \cite{Amari2016}, Theorem 1.4) Given $%
\alpha \in \mathcal{M}$ and a smooth submanifold $\mathcal{S}$ in a dually
flat manifold $\mathcal{M}$, the vector that minimises the divergence $%
D_{\varphi }\left[ \alpha :\beta \right] $, $\beta \in \mathcal{S}$, is the
dual geodesic projection of $\alpha $ to $\mathcal{S}$. The vector that
minimises the dual divergence $D_{\varphi ^{\times }}\left[ \alpha :\beta %
\right] $, $\beta \in \mathcal{S}$, is the geodesic projection of $\alpha $
to $\mathcal{S}.$
\end{Theo}

\subsection{Problem formulation}

The main intention of introducing the system coprime factorisations and the
SIR/SKR concepts is to model the (nominal) system dynamic as a
subspace/sub-manifold in the data space of $\left( u,y\right) .$ Although
these model forms are equivalent to the input-output models widely applied
in control theoretic research, they provide us with the alternative
methodology to deal with fault diagnosis problems, namely directly in data
space and by means of an orthogonal projection. As delineated by the
projection-based fault detection scheme for LTI systems \cite{ding2022}, an
orthogonal projection of $\left( u,y\right) $ onto the system image (kernel)
subspace helps us to extract information from the data optimally using the
distance of the data vector to the system image subspace as the (nominal)
system dynamics. In this regard, the normalised SIR and SKR serve for two
purposes. Firstly, it is a mathematical tool for the computation of the
projection and the distance. Secondly, it connects the projection-based
fault detection with the observer-based one, as expressed by (\ref{eq2-20})
and, in particular, (\ref{eq2-21a}). The latter implies that the
projection-based fault detection can be realised using an observer-based
detection system and thus, on the other hand, enables an optimal design of
an observer-based detection system.

\bigskip

A direct extension of the projection-based fault detection to nonlinear
dynamic systems poses obviously a number of problems. While the SIR and SKR
can be, analogue to LTI systems, straightforward introduced, concepts like
projection and distance definitions have to be addressed with the aid of
different mathematical tools and knowledge. Concerning projection issues,
the well-established theoretic framework of Hamiltonian systems \cite%
{Scherpen1994,Ball1996,vanderSchaftbook} may serve as the system foundation,
and as an analysis and computational tool, where the Hamiltonian systems
corresponding to the normalised SIR and SKR and their solutions as inner and
co-inner systems play an important role in our subsequent work. The
centerpiece in our framework is the associated Hamiltonian functions, by
which the Bregman divergences are introduced for our fault detection study.

\bigskip

To be specific, for LTI systems, the system image (kernel) is a subspace in
Hilbert space, in which an inner product and the corresponding norm of a
vector are well defined. In such a Euclidean space, the shortest distance
between two points (vectors) is a straight line, and thus the norm of a
residual vector is efficient for measuring the distance from a data vector
to the system image subspace that represents the nominal dynamics. This is
the theoretic basis for an optimal fault detection. In a non-Euclidean
space, as the manifold formed by process data from a nonlinear system, the
metric tensor at each point of the manifold is different from the Euclidean
metric tensor like an inner product defined on Hilbert space. Therefore, the
norm-based distance defined in Hilbert space may not be compatible with the
metric tensor of the non-Euclidean space. In other words, the norm-based
metric is not suitable to measure the distance from a data vector to the
system image manifold, based on which a fault detection should be \ \ made.
To solve this problem, we propose to use Bregman divergences to measure
deviations in the system dynamics and data from the nominal ones. Initially,
this idea has been inspirited by wide applications of Bregman divergences in
machine learning \cite{Adamcik2014,FSG2008} and, in particular, the
Kullback-Leibler (KL) divergence as a difference measure of two probability
distributions. In course of our extensive investigation, three
distinguishing aspects of applying Bregman divergences to fault detection
draw our attention : (i) realisation of performance-oriented fault
detection, (ii) simplified computation and realisation of Bregman
divergences making use of the dual Hamiltonian system pair as a Legendre
transform,\ and (iii) interpretation of the SIR-based projection as a
geodesic projection in terms of Bregman divergence from a data vector to the
image manifold. To be specific,

\begin{itemize}
\item performance-oriented fault detection has been proposed in our earlier
works \cite{LLDYP-2019,LDLPY-TII-2020,Ding2020}, dedicated to reliably
detecting performance degradation in automatic control systems. It has been
observed that considerable engineering efforts are often needed to set the
parameters of the performance evaluation function. Bregman divergences are
induced by convex functions and measure the difference of two points
(vectors) defined in terms of this convex function. In optimal control,
Hamiltonian functions act as (control) performance cost functions. In a
broad sense, Hamiltonian functions can be interpreted as the energy balance
(parity relation) in dynamic systems. This fact motivates us to define the
Hamiltonian function of the system under consideration as a system
performance and, based on it, to induce Bregman divergences towards a
performance-oriented fault detection.

\item It is known that for a given system, there exists a dual Hamiltonian
system pair that forms a Legendre \ transform. As delineated in Subsection %
\ref{subsec2-4}, the Legendre dual of the Hamiltonian functions is useful
for the Bregman divergence computations. More importantly, the Bregman
divergence endows the data manifold with a dually flat structured Riemannian
manifold, which can be regarded as a dualistic extension of the Euclidean
space. And thus,

\item the projection theorem \cite{Amari2016}, Theorem \ref{Theo2-2}, can be
used to characterise the Hamiltonian function-based projection and to
highlight its insight from the viewpoint of information geometry.
\end{itemize}

As described in Subsection 2.1, a distinguishing feature for projections in
Hilbert space is that the vector can be decomposed as the orthogonal
projection of the data vector onto the image subspace and its complementary.
It leads to the useful relations (\ref{eq2-20}) and (\ref{eq2-21a}) that
enable to realise a projection-based fault detection system in the
well-established observer-based detection framework (refer to (\ref{eq2-21a}%
), the Pythagorean equation). Moreover, they pose a novel scheme to design
observer-based fault detection systems that are capable to deal with systems
with model uncertainties and multiplicative faults \cite{ding2022}.
Unfortunately, for nonlinear systems such relations are generally not valid.
Bearing this fact in mind, we will study the (nonlinear) SIR- and SKR-based
projections separately. The SIR-based projection, similar to the application
to LTI systems, concerns with the projection of process data $\left(
u,y\right) $ to the system image manifold representing the nominal dynamics.
By means of a Hamiltonian function induced Bregman divergence, the
projection is evaluated, which delivers a distance measure between the data
and the image manifold and is used for fault detection. As described in
Subsections \ref{sec2-1} and 2.2, an SKR is an observer-based residual
generator that reflects influences of uncertainties in the system, including
faults, on \ the system dynamics. Correspondingly, the SKR-based projection
is an estimate for uncertainties. For the fault detection purpose, the
evaluation of the SKR-based projection in terms of a Hamiltonian function
induced Bregman divergence will be dedicated to a measure of the distance
between the data and the uncertainty-free operation.

\bigskip

In course of our study on the SKR-based projection, the imperative of
insightfully understanding the SKR-based projection as an estimate for
uncertainties has been recognised. Remember that the essential reason why
the relation (\ref{eq2-21a}) doesn't hold for nonlinear systems is that the
data $\left( u,y\right) $ are inextricably corrupted with uncertainties.
Consequently, the influences of the system input (representing the nominal
dynamic) and uncertainties in the system are not separable in the data. As
will be shown in Section 4, an SKR-based projection is a Hamiltonian system
with the system data $\left( u,y\right) $ (equivalent with an observer-based
residual generator) as its input and the projection as its output. Our
endeavour, to be specific, is to figure out in which context the projection
delivers an estimate of the uncertainty corrupted in the data. To this end,
an uncertainty model is introduced, which uniformly models deviations in the
data from their nominal values and caused by all types of possible
uncertainties, including faults. In this regard, it will be delineated that
the SKR-based projection is an LS estimate of the uncertainty corrupted in
the data. This result not only highlights the insightful interpretation of
the SKR-based projection, but also underlines its potential use, for
instance, for fault estimation.

\bigskip

To put it in a nutshell, the main problems to be addressed in this work are
formulated as follows:

\begin{itemize}
\item configuration of a normalised SIR-based projection system that
projects $\left( u,y\right) $ onto $\mathcal{I}_{\Sigma }.$ The work will be
performed in the framework of Hamiltonian systems and their dual form as a
Legendre transform.

\item introduction of a Bregman divergence induced by a Hamiltonian function
and computation solution of the Bregman divergence from $\left( u,y\right) $
to its projection. This solution will serve as an evaluation function for
the detection (decision) logic.

\item study on interpretations of the SIR-based projection and Bregman
divergence as a fault detection system. Hereby, two aspects are of special
interests: (i) the control theoretic insight, i.e. understanding of the
Hamiltonian function induced Bregman divergence in the context of system
performance-oriented fault detection, (ii) the interpretation from the
viewpoint of information geometry and in the context of the projection
theorem.

\item discussion on the realisation and implementation issues. In order to
achieve a reliable fault detection, it is of considerable practical
importance to reduce false alarms. To this end, the defined Bregman
divergence will be evaluated over an evaluation time interval. Moreover, a
threshold should be set to ensure an optimal fault detection decision. At
the end of this work, a one-class fault detection scheme for nonlinear
dynamic systems is completed.

\item investigation on a scheme of a normalised SKR-based projection system.
Although this work is analogue to the one of the SIR-based projection and
fault detection system, it is significantly different in (i) its
construction as an observer-based residual generator, (ii) the definition of
the Bregman divergence, and in particular (iii) the \ interpretation of the
Bregman divergence as a distance measure of $\left( u,y\right) $ and the
nominal system dynamics.

\item study on the SKR-based projection system as an estimator for the
uncertainty corrupted in $\left( u,y\right) $ with a focus on (i) definition
of an uncertainty model that should be independent of the possible types of
uncertainties in the system, including faults, and thus lays the basis for
interpreting and evaluating the estimation performance, (ii) evaluation of
the projection-based estimation.
\end{itemize}

\section{A projection and Bregman divergence based fault detection scheme}

In this section, the first part of our work is presented. It is a
projection-based scheme for detecting faults in the affine system (\ref%
{eq2-1}). To be specific, a projection of process data $\left( u,y\right) $
onto the image subspace $\mathcal{I}_{\Sigma }$ is realised using the
normalised SIR $\Sigma _{\mathcal{I}}.$ Then, the introduction of a
Hamiltonian function induced Bregman divergence and the associated
Riemannian structure of the data manifold enables the definition and
computation of a geodesic projection onto $\mathcal{I}_{\Sigma }$. Based on
these theoretical results, fault detection algorithms are presented. To this
end, we will first study the normalised SIR and its properties in the
framework of Hamiltonian systems. It is followed by defining and analysing a
Bregman divergence with the Hamiltonian function as a generating function.
The proof that the normalised SIR-based projection is the geodesic
projection gives an interesting interpretation of the projection as an
important result of this study. At the end of this section, implementation
issues towards a practical fault detection and the associated algorithm are
discussed.

\subsection{Normalised SIR and the associated Hamiltonian projection \
systems}

Consider the affine system (\ref{eq2-1}) and its SIR $\Sigma _{\mathcal{I}}$
given by (\ref{eq2-2}). The corresponding Hamiltonian extension of $\Sigma _{%
\mathcal{I}}$ is the dynamic system described by%
\begin{equation*}
\begin{cases}
\dot{x}=a_{I}(x)+B_{I}(x)v \\ 
\dot{\lambda}=-\left( \frac{\partial a_{I}(x)}{\partial x}+\frac{\partial
B_{I}(x)}{\partial x}v\right) ^{T}\lambda -\left( \frac{\partial c_{I}\left(
x\right) }{\partial x}+\frac{\partial D_{I}(x)}{\partial x}\right) ^{T}v_{a}
\\ 
\hat{z}=c_{I}(x)+D_{I}(x)v\in \mathbb{\ R}^{m+p} \\ 
z_{a}:=B_{I}^{T}(x)\lambda +D_{I}^{T}(x)v_{a}\in \mathbb{R}^{p}.%
\end{cases}%
\end{equation*}%
For our purpose of projecting $\left( u,y\right) $ onto the image subspace $%
\mathcal{I}_{\Sigma },$ connect $v$ and $z_{a},$ i.e. 
\begin{equation}
v=z_{a}=B_{I}^{T}(x)\lambda +D_{I}^{T}(x)v_{a},  \label{eq3-3}
\end{equation}%
and let the data $\left( u,y\right) $ be the input and $\hat{z}$ the output
of $\Sigma _{\mathcal{I}}\circ \left( D\Sigma _{\mathcal{I}}\right) ^{T},$ 
\begin{equation*}
v_{a}=\left[ 
\begin{array}{c}
u \\ 
y%
\end{array}%
\right] =:z,\hat{z}=\Sigma _{\mathcal{I}}\circ \left( D\Sigma _{\mathcal{I}%
}\right) ^{T}\left( z\right) .
\end{equation*}%
Then, we have 
\begin{gather}
\Sigma _{\mathcal{I}}\circ \left( D\Sigma _{\mathcal{I}}\right) ^{T}:%
\begin{cases}
\dot{x}=a_{I}(x)+B_{I}(x)B_{I}^{T}(x)\lambda +B_{I}(x)D_{I}^{T}(x)z \\ 
\dot{\lambda}=-\left( \frac{\partial a_{I}(x)}{\partial x}+\frac{\partial 
\bar{B}_{I}\left( x,\lambda ,z\right) }{\partial x}\right) ^{T}\lambda
-\left( \frac{\partial \left( c_{I}(x)+\bar{D}_{I}\left( x,\lambda ,z\right)
\right) }{\partial x}\right) ^{T}z \\ 
\hat{z}=c_{I}(x)+D_{I}(x)B_{I}^{T}(x)\lambda +D_{I}(x)D_{I}^{T}(x)z,%
\end{cases}
\label{eq3-1} \\
\bar{B}_{I}\left( x,\lambda ,z\right) =B_{I}(x)B_{I}^{T}(x)\lambda
+B_{I}(x)D_{I}^{T}(x)z,  \notag \\
\bar{D}_{I}\left( x,\lambda ,z\right) =D_{I}(x)B_{I}^{T}(x)\lambda
+D_{I}(x)D_{I}^{T}(x)z.  \notag
\end{gather}%
The configuration of $\Sigma _{\mathcal{I}}\circ \left( D\Sigma _{\mathcal{I}%
}\right) ^{T}$ is schematically depicted in Figure \ref{Fig3-1}. 
\begin{figure}[h]
\centering\includegraphics[width=9cm,height=2cm]{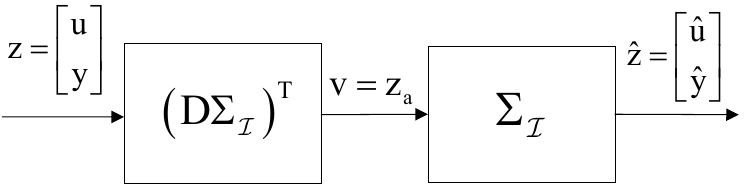}
\caption{SIR-based projection system $\Sigma _{\mathcal{I}}\circ \left(
D\Sigma _{\mathcal{I}}\right) ^{T}$}
\label{Fig3-1}
\end{figure}

Define the Hamiltonian function%
\begin{equation}
H\left( x,\lambda ,z\right) =\frac{1}{2}%
z^{T}D_{I}(x)D_{I}^{T}(x)z+c_{I}^{T}(x)z+\lambda ^{T}\left( a_{I}(x)+\frac{1%
}{2}B_{I}(x)B_{I}^{T}(x)\lambda +B_{I}(x)D_{I}^{T}(x)z\right) .
\label{eq3-2}
\end{equation}%
It is straightforward that the Hamiltonian system (\ref{eq3-1}) can then be
written in the compact form%
\begin{equation}
\Sigma _{\mathcal{I}}\circ \left( D\Sigma _{\mathcal{I}}\right) ^{T}:%
\begin{cases}
\dot{x}=\frac{\partial H}{\partial \lambda }\left( x,\lambda ,z\right) \\ 
\dot{\lambda}=-\frac{\partial H}{\partial x}\left( x,\lambda ,z\right) \\ 
\hat{z}=\frac{\partial H}{\partial z}\left( x,\lambda ,z\right) .%
\end{cases}
\label{eq3-1a}
\end{equation}%
Notice that the Hamiltonian function (\ref{eq3-2}) can be, after some
calculations, written as%
\begin{align}
H\left( x,\lambda ,z\right) & =-\frac{1}{2}z^{T}D_{I}(x)D_{I}^{T}(x)z+\hat{z}%
^{T}z+\lambda ^{T}\left( a_{I}(x)+\frac{1}{2}B_{I}(x)B_{I}^{T}(x)\lambda
\right)  \notag \\
& =-\frac{1}{2}v^{T}v+\hat{z}^{T}z+\lambda ^{T}\dot{x},  \label{eq3-2a} \\
\dot{x}& =a_{I}(x)+B_{I}(x)B_{I}^{T}(x)\lambda +B_{I}(x)D_{I}^{T}(x)z. 
\notag
\end{align}%
Let 
\begin{equation}
H^{\times }\left( x,\lambda ,\hat{z}\right) =\hat{z}^{T}z-H\left( x,\lambda
,z\right) =\frac{1}{2}v^{T}v-\lambda ^{T}\dot{x}.  \label{eq3-2b}
\end{equation}%
Note that 
\begin{equation*}
z=\frac{\partial H^{\times }\left( x,\lambda ,\hat{z}\right) }{\partial \hat{%
z}},
\end{equation*}%
and the Hamiltonian system (\ref{eq3-1a}) can be expressed in terms of $%
H^{\times }\left( x,\lambda ,\hat{z}\right) $ as%
\begin{equation}
\begin{cases}
\dot{x}=-\frac{\partial H^{\times }}{\partial \lambda }\left( x,\lambda ,%
\hat{z}\right) \\ 
\dot{\lambda}=\frac{\partial H^{\times }}{\partial x}\left( x,\lambda ,\hat{z%
}\right) \\ 
z=\frac{\partial H^{\times }}{\partial \hat{z}}\left( x,\lambda ,\hat{z}%
\right)%
\end{cases}
\label{eq3-1b}
\end{equation}%
with $\hat{z}$ as the input and $z$ as the output. The above Hamiltonian
system (\ref{eq3-1b}) has been investigated in the context of inner-outer
factorisation and can be interpreted as the inverse of $-\Sigma _{\mathcal{I}%
}\circ \left( D\Sigma _{\mathcal{I}}\right) ^{T}$ and $H^{\times }\left(
x,\lambda ,\hat{z}\right) $ as the Legendre dual of $H\left( x,\lambda
,z\right) $ \cite{Ball1996}. In other words, the Hamiltonian functions $%
H\left( x,\lambda ,z\right) $ and $H^{\times }\left( x,\lambda ,\hat{z}%
\right) $ with%
\begin{equation*}
\hat{z}=\frac{\partial H}{\partial z}\left( x,\lambda ,z\right) =:z^{\times
},z=\frac{\partial H^{\times }}{\partial \hat{z}}\left( x,\lambda ,\hat{z}%
\right) ,
\end{equation*}%
form a Legendre transform and give the system data a dualistic structure.
Here, $\hat{z}$ is dual to $z$ and thus is also denoted by $z^{\times }.$ It
is obvious from (\ref{eq3-2a}) that 
\begin{equation}
H^{\times }\left( x,\lambda ,z^{\times }\right) =\frac{1}{2}v^{T}v-\lambda
^{T}\dot{x}=\frac{1}{2}z^{T}D_{I}(x)D_{I}^{T}(x)z-\lambda ^{T}\left(
a_{I}(x)+\frac{1}{2}B_{I}(x)B_{I}^{T}(x)\lambda \right) .  \label{eq3-4}
\end{equation}%
Now, consider the normalised SIR and the corresponding Hamiltonian systems (%
\ref{eq3-1a}) and (\ref{eq3-1b}). The \ following theorem showcases that $%
\Sigma _{\mathcal{I}}\circ \left( D\Sigma _{\mathcal{I}}\right) ^{T}$ is an
idempotent operator and the normalised SIR projects the process data $\left(
u,y\right) $ to $\mathcal{I}_{\Sigma }.$

\begin{Theo}
\label{Theo3-1}Given the normalised SIR and the corresponding Hamiltonian
system (\ref{eq3-1}), then $\Sigma _{\mathcal{I}}\circ \left( D\Sigma _{%
\mathcal{I}}\right) ^{T}$ is idempotent, and%
\begin{align}
\forall z& =\left[ 
\begin{array}{c}
u \\ 
y%
\end{array}%
\right] \in \mathcal{L}_{2},\hat{z}\in \mathcal{I}_{\Sigma },  \label{eq3-3a}
\\
\forall z& \neq 0,\hat{z}=\left[ 
\begin{array}{c}
u \\ 
y%
\end{array}%
\right] \text{ iff }\left[ 
\begin{array}{c}
u \\ 
y%
\end{array}%
\right] \in \mathcal{I}_{\Sigma }.  \label{eq3-3b}
\end{align}
\end{Theo}

\begin{proof}
That $\Sigma _{\mathcal{I}}\circ \left( D\Sigma _{\mathcal{I}}\right) ^{T}$
is idempotent immediately follows from the normalised SIR, i.e. 
\begin{equation*}
\Sigma _{\mathcal{I}}\circ \left( D\Sigma _{\mathcal{I}}\right) ^{T}\circ
\Sigma _{\mathcal{I}}\circ \left( D\Sigma _{\mathcal{I}}\right) ^{T}=\Sigma
_{\mathcal{I}}\circ \left( D\Sigma _{\mathcal{I}}\right) ^{T}.
\end{equation*}%
The proof of (\ref{eq3-3a}) is apparent, since 
\begin{equation*}
\left( D\Sigma _{\mathcal{I}}\right) ^{T}\left( \left[ 
\begin{array}{c}
u \\ 
y%
\end{array}%
\right] \right) =:v\in \mathcal{L}_{2}\Longrightarrow \Sigma _{\mathcal{I}%
}\circ \left( D\Sigma _{\mathcal{I}}\right) ^{T}\left( \left[ 
\begin{array}{c}
u \\ 
y%
\end{array}%
\right] \right) =\Sigma _{\mathcal{I}}\left( v\right) \in \mathcal{I}%
_{\Sigma }.
\end{equation*}%
To prove (\ref{eq3-3b}), observe that 
\begin{gather*}
\forall \left[ 
\begin{array}{c}
u \\ 
y%
\end{array}%
\right] \in \mathcal{I}_{\Sigma },\exists v\in \mathcal{L}_{2},\text{ s.t. }%
\left[ 
\begin{array}{c}
u \\ 
y%
\end{array}%
\right] =\Sigma _{\mathcal{I}}\left( v\right) \Longrightarrow  \\
\Sigma _{\mathcal{I}}\circ \left( D\Sigma _{\mathcal{I}}\right) ^{T}\left( %
\left[ 
\begin{array}{c}
u \\ 
y%
\end{array}%
\right] \right) =\Sigma _{\mathcal{I}}\circ \left( D\Sigma _{\mathcal{I}%
}\right) ^{T}\circ \Sigma _{\mathcal{I}}\left( v\right) =\Sigma _{\mathcal{I}%
}\left( v\right) =\left[ 
\begin{array}{c}
u \\ 
y%
\end{array}%
\right] .
\end{gather*}%
On the other hand, it follows from (\ref{eq3-3a}) that 
\begin{equation*}
\hat{z}=\left[ 
\begin{array}{c}
u \\ 
y%
\end{array}%
\right] \Longrightarrow \left[ 
\begin{array}{c}
u \\ 
y%
\end{array}%
\right] \in \mathcal{I}_{\Sigma }.
\end{equation*}%
Thus, (\ref{eq3-3b}) is proved. 
\end{proof}

\begin{Rem}
That $\Sigma _{\mathcal{I}}\circ \left( D\Sigma _{\mathcal{I}}\right) ^{T}$
is idempotent is an important property that ensures the projection
preserving local neighbourhoods at every point of the underlying manifold.
\end{Rem}

As a projection of $\left( u,y\right) $ to $\mathcal{I}_{\Sigma },$ $\hat{z}$
can be interpreted as an estimate of $\left( u,y\right) $ as well. In this
context, the notation 
\begin{equation*}
\hat{z}=\left[ 
\begin{array}{c}
\hat{u} \\ 
\hat{y}%
\end{array}%
\right]
\end{equation*}%
will also be adopted in the sequel.

\bigskip

Next, we consider the Hamiltonian functions (\ref{eq3-2}) and (\ref{eq3-4})
in case of the normalised SIR. It follows from Lemma \ref{Lemma2-1} that for
the normalised SIR, 
\begin{equation*}
P_{x}^{T}\left( x\right) a_{I}(x)+\frac{1}{2}%
c_{I}^{T}(x)c_{I}(x)=0,P_{x}^{T}\left( x\right)
B_{I}(x)+c_{I}^{T}(x)D_{I}(x)=0,D_{I}^{T}(x)D_{I}(x)=I,
\end{equation*}%
which yields%
\begin{gather*}
\frac{1}{2}\left( c_{I}(x)+D_{I}(x)B_{I}^{T}(x)P_{x}\left( x\right) \right)
^{T}\left( c_{I}(x)+D_{I}(x)B_{I}^{T}(x)P_{x}\left( x\right) \right) \\
+P_{x}^{T}\left( x\right) \left( a_{I}(x)+\frac{1}{2}%
B_{I}(x)B_{I}^{T}(x)P_{x}\left( x\right) \right) = \\
P_{x}^{T}\left( x\right) a_{I}(x)+\frac{1}{2}c_{I}^{T}(x)c_{I}(x)+P_{x}^{T}%
\left( x\right) B_{I}(x)B_{I}^{T}(x)P_{x}\left( x\right)
+c_{I}^{T}(x)D_{I}(x)B_{I}^{T}(x)P_{x}\left( x\right) =0.
\end{gather*}%
Hence, it holds 
\begin{align}
& P_{x}^{T}\left( x\right) \left( a_{I}(x)+\frac{1}{2}B_{I}(x)B_{I}^{T}(x)%
\lambda \right) -\frac{1}{2}z^{T}D_{I}(x)D_{I}^{T}(x)z  \notag \\
& =-\frac{1}{2}\left( c_{I}(x)+D_{I}(x)B_{I}^{T}(x)P_{x}\left( x\right)
\right) ^{T}\left( c_{I}(x)+D_{I}(x)B_{I}^{T}(x)P_{x}\left( x\right) \right)
-\frac{1}{2}z^{T}D_{I}(x)D_{I}^{T}(x)z  \notag \\
& =-\frac{1}{2}\left( \hat{z}-D_{I}(x)D_{I}^{T}(x)z\right) ^{T}\left( \hat{z}%
-D_{I}(x)D_{I}^{T}(x)z\right) -\frac{1}{2}z^{T}D_{I}(x)D_{I}^{T}(x)z  \notag
\\
& =-\frac{1}{2}\hat{z}^{T}\hat{z}+\hat{z}%
^{T}D_{I}(x)D_{I}^{T}(x)z-z^{T}D_{I}(x)D_{I}^{T}(x)z.  \label{eq3-5d}
\end{align}%
In the course of this computation we have used the relations%
\begin{gather*}
\hat{z}=c_{I}(x)+D_{I}(x)B_{I}^{T}(x)P_{x}\left( x\right)
+D_{I}(x)D_{I}^{T}(x)z, \\
D_{I}(x)D_{I}^{T}(x)D_{I}(x)D_{I}^{T}(x)=D_{I}(x)D_{I}^{T}(x).
\end{gather*}

By (\ref{eq3-2a}), (\ref{eq3-4}) and (\ref{eq3-5d}), the following theorem
is apparent.

\begin{Theo}
\label{Theo3-2}Given the normalised SIR and the corresponding Hamiltonian
system (\ref{eq3-1}), then%
\begin{align}
H^{\times }\left( x,\lambda ,z^{\times }\right) & =\frac{1}{2}\hat{z}^{T}%
\hat{z},  \label{eq3-6a} \\
H\left( x,\lambda ,z\right) & =\hat{z}^{T}z-\frac{1}{2}\hat{z}^{T}\hat{z}.
\label{eq3-6b}
\end{align}
\end{Theo}

\begin{proof}
It is apparent that we have, according to (\ref{eq3-2a}), (\ref{eq3-4}) and (%
\ref{eq3-5d}), 
\begin{align*}
H^{\times }\left( x,\lambda ,z^{\times }\right) & =\frac{1}{2}\hat{z}^{T}%
\hat{z}-\hat{z}^{T}D_{I}(x)D_{I}^{T}(x)z+z^{T}D_{I}(x)D_{I}^{T}(x)z, \\
H\left( x,\lambda ,z\right) & =\hat{z}^{T}\left(
I+D_{I}(x)D_{I}^{T}(x)\right) z-\frac{1}{2}\hat{z}^{T}\hat{z}%
-z^{T}D_{I}(x)D_{I}^{T}(x)z.
\end{align*}%
Notice that 
\begin{align*}
\hat{z}^{T}D_{I}(x)D_{I}^{T}(x)z& =\left(
c_{I}(x)+D_{I}(x)B_{I}^{T}(x)P_{x}\left( x\right)
+D_{I}(x)D_{I}^{T}(x)z\right) ^{T}D_{I}(x)D_{I}^{T}(x)z \\
& =\left( c_{I}^{T}(x)D_{I}(x)+P_{x}^{T}\left( x\right) B_{I}(x)\right)
D_{I}^{T}(x)z+z^{T}D_{I}(x)D_{I}^{T}(x)z \\
& =z^{T}D_{I}(x)D_{I}^{T}(x)z.
\end{align*}%
Hence, 
\begin{equation*}
H^{\times }\left( x,\lambda ,z^{\times }\right) =\frac{1}{2}\hat{z}^{T}\hat{z%
},H\left( x,\lambda ,z\right) =\hat{z}^{T}z-\frac{1}{2}\hat{z}^{T}\hat{z}.
\end{equation*}
\end{proof}

An immediate result of Theorems \ref{Theo3-1} and \ref{Theo3-2} is the
following corollary, which is useful for us to examine the interpretation of
the both Hamiltonian functions as generalised energy functions and, based on
it, to establish our fault detection strategy.

\begin{Corol}
\label{col3-1}Given the normalised SIR and the corresponding Hamiltonian
system (\ref{eq3-1}), then%
\begin{equation*}
\forall \left[ 
\begin{array}{c}
u \\ 
y%
\end{array}%
\right] \in \mathcal{I}_{\Sigma },H^{\times }\left( x,\lambda ,z^{\times
}\right) =H\left( x,\lambda ,z\right) =\frac{1}{2}\hat{z}^{T}\hat{z}=\frac{1%
}{2}z^{T}z=\frac{1}{2}v^{T}v-\dot{P}(x),
\end{equation*}%
where $P(x)\geq 0$ is the solution of the HJE (\ref{eq2-15a}).
\end{Corol}

\begin{proof}
It follows from Theorem \ref{Theo3-1} that 
\begin{equation*}
\forall \left[ 
\begin{array}{c}
u \\ 
y%
\end{array}%
\right] \in \mathcal{I}_{\Sigma },\hat{z}=\Sigma _{\mathcal{I}}\circ \left(
D\Sigma _{\mathcal{I}}\right) ^{T}\left( z\right) =z.
\end{equation*}%
Thus, according to Theorem \ref{Theo3-2} that 
\begin{equation*}
H^{\times }\left( x,\lambda ,z^{\times }\right) =H\left( x,\lambda ,z\right)
=\frac{1}{2}\hat{z}^{T}\hat{z}=\frac{1}{2}z^{T}z.
\end{equation*}%
On the other hand, for the normalised SIR, 
\begin{equation*}
\lambda =P_{x}(x)\Longrightarrow \lambda ^{T}\dot{x}=\dot{P}(x),
\end{equation*}%
which leads to 
\begin{gather*}
H\left( x,\lambda ,z\right) =\hat{z}^{T}z-\frac{1}{2}v^{T}v+\dot{P}%
(x)=H^{\times }\left( x,\lambda ,z^{\times }\right) =\frac{1}{2}v^{T}v-\dot{P%
}(x) \\
\Longrightarrow \hat{z}^{T}z=z^{T}z=v^{T}v-2\dot{P}(x)\Longrightarrow
H\left( x,\lambda ,z\right) =H^{\times }\left( x,\lambda ,z^{\times }\right)
=\frac{1}{2}v^{T}v-\dot{P}(x).
\end{gather*}%
The theorem is proved. 
\end{proof}

\bigskip

Corollary \ref{col3-1} tells us, corresponding to the nominal system
dynamics, both Hamiltonian functions (\ref{eq3-2}) and (\ref{eq3-2b}) as
energy functions are equivalent. Moreover, from the viewpoint of the energy
balance, the lossless property, as defined in Definition \ref{Def2-2} and in
the sense of (\ref{eq2-11a}), holds. When uncertainties or faults are
present in the system, i.e. $z\notin \mathcal{I}_{\Sigma },$ such an energy
balance will be disturbed, as described by (\ref{eq3-6b}) in Theorem \ref%
{Theo3-2}. In this context, the fault detection problem can be formulated as
detecting deviations of the Hamiltonian functions from the nominal value $%
\frac{1}{2}\hat{z}^{T}\hat{z}=\frac{1}{2}z^{T}z.$ To this end, it is plain
that a well-founded evaluation function should be defined in such a way that
it is possible to examine the energy balance relations real-time and thus to
detect possible changes in the parity relations reliably and timely. This is
the problem to be addressed in the next subsection.

\subsection{A Bregman divergence-based evaluation function, and geodesic
projection}

In the well-established observer-based fault detection framework, building
(i) a residual vector as the difference between the measurement and its
estimate, (ii) evaluating the residual vector by means of a signal norm,
typically $\mathcal{L}_{2}$ norm, are the two essential steps to a
successful fault detection. In our divergence-based detection framework, we
propose, alternatively, to evaluate possible changes in the system dynamic
in terms of a divergence between the measurement (data) and its estimate
(projection). The basic idea behind this approach is twofold, as motivated
and described in the section of problem formulation,

\begin{itemize}
\item to achieve a system performance-oriented fault detection by means of a
Hamiltonian function induced Bregman divergence, and

\item to detect deviations using a Bregman divergence as a distance measure.
\end{itemize}

Consider the normalised SIR (\ref{eq2-2}) and the corresponding Hamiltonian
system (\ref{eq3-1}) with the input and output vectors, $z$ and $\hat{z},$
where $z$ is the collected data and $\hat{z}$ is its estimate (projection)
delivered by (\ref{eq3-1}). In the sequel, for the sake of notation
simplicity, the Hamiltonian functions (\ref{eq3-2}) and (\ref{eq3-4}) are
expressed respectively by%
\begin{equation*}
H(z):=H\left( x,\lambda ,z\right) ,H^{\times }(z):=H^{\times }\left(
x,\lambda ,z\right) .
\end{equation*}%
Since our objective is to detect faults by checking deviations from the
nominal dynamics, we examine, according to Corollary \ref{col3-1}, the
condition%
\begin{equation}
H^{\times }(z)=H(z)=\frac{1}{2}\hat{z}^{T}\hat{z}=\frac{1}{2}z^{T}z.
\label{eq3-11a}
\end{equation}%
To this end, the following hypothesis test is under consideration:%
\begin{equation}
\left\{ 
\begin{array}{l}
\mathcal{H}_{0}:H^{\times }(z)=H(z)=\frac{1}{2}\hat{z}^{T}\hat{z}=\frac{1}{2}%
z^{T}z \\ 
\mathcal{H}_{1}:H(z)\neq \frac{1}{2}z^{T}z.%
\end{array}%
\right.  \label{eq3-11}
\end{equation}%
Here, the hypothesis $\mathcal{H}_{0}$ represents the nominal dynamics,
while $\mathcal{H}_{1}$ indicates faulty dynamics. Let the Bregman
divergence 
\begin{equation*}
D_{H_{0}}\left[ \hat{z}:z\right] :=H_{0}(\hat{z})-H_{0}(z)-\left( \frac{%
\partial H_{0}(z)}{\partial z}\right) ^{T}\left( \hat{z}-z\right) ,
\end{equation*}%
be the test (evaluation) function for the implementation of the hypothesis
test (\ref{eq3-11}), where $H_{0}(z)$ represents the nominal dynamics and
thus is subject to 
\begin{equation*}
H_{0}(z)=\frac{1}{2}z^{T}z.
\end{equation*}%
The computation of $D_{H_{0}}\left[ \hat{z}:z\right] $ is given in following
theorem, which also provides us with an interesting interpretation of $%
D_{H_{0}}\left[ \hat{z}:z\right] .$

\begin{Theo}
\label{Theo3-3}Given the normalised SIR (\ref{eq2-2}), the corresponding
Hamiltonian system (\ref{eq3-1}) and Hamiltonian functions (\ref{eq3-2}) and
(\ref{eq3-4}), then the Bregman divergence 
\begin{equation*}
D_{H_{0}}\left[ \hat{z}:z\right] =H_{0}(\hat{z})-H_{0}(z)-\left( \frac{%
\partial H_{0}(z)}{\partial z}\right) ^{T}\left( \hat{z}-z\right)
\end{equation*}%
is given by 
\begin{align}
D_{H_{0}}\left[ \hat{z}:z\right] & =H_{0}\left( z\right) -H\left( z\right) ,
\label{eq3-12} \\
H\left( z\right) & =\hat{z}^{T}z-\frac{1}{2}\hat{z}^{T}\hat{z}.
\end{align}
\end{Theo}

\begin{proof}
Recall that $z$ and $\hat{z}$ are a Legendre dual pair, 
\begin{equation*}
\hat{z}=\frac{\partial H\left( x,\lambda ,z\right) }{\partial z}=z^{\times
},z=\frac{\partial H^{\times }\left( x,\lambda ,\hat{z}\right) }{\partial 
\hat{z}}.
\end{equation*}%
Hence, according to Corollary \ref{col3-1}, (\ref{eq3-6b}) in Theorem \ref%
{Theo3-2} and (\ref{eq2-22}), we have%
\begin{eqnarray*}
D_{H_{0}}\left[ \hat{z}:z\right]  &=&H_{0}(\hat{z})+H_{0}\left( z\right)
-z^{T}\hat{z}=H_{0}\left( z\right) -H\left( z\right) , \\
H_{0}\left( z\right)  &=&\frac{1}{2}z^{T}z,H\left( z\right) =z^{T}\hat{z}%
-H_{0}(\hat{z})=\hat{z}^{T}z-\frac{1}{2}\hat{z}^{T}\hat{z}.
\end{eqnarray*}
\end{proof}

\bigskip

Theorem \ref{Theo3-3} demonstrates that the Bregman divergence $D_{H_{0}}%
\left[ \hat{z}:z\right] $ gives a distance measure between process data $z$
and its projection on to $\mathcal{I}_{\Sigma },\hat{z},$ in terms of the
deviation of the Hamiltonian function from its nominal value $H_{0}(z)$,
expressed by $H_{0}\left( z\right) -H\left( z\right) .$

\bigskip

We would like to draw the reader's attention to the fact that 
\begin{equation*}
D_{H_{0}}\left[ \hat{z}:z\right] =\frac{1}{2}z^{T}z+\frac{1}{2}\hat{z}^{T}%
\hat{z}-\hat{z}^{T}z=\frac{1}{2}\left\Vert z-\hat{z}\right\Vert ^{2},
\end{equation*}%
i.e. $D_{H_{0}}\left[ \hat{z}:z\right] $ is one half of the squared
Euclidean norm of the projection-based residual $z-\hat{z}.$ Although, due
to the generating function defined by $H_{0}(z)=\frac{1}{2}z^{T}z,$ this is
a known mathematical result, the underlined control-theoretic interpretation
is of remarkable importance. That is, the squared Euclidean norm of the
projection-based residual gives the deviation of the Hamiltonian function
from its nominal value. In other words, the norm-based evaluation of the
projection-based residual implies a performance-oriented evaluation and thus
leads to a performance-oriented fault detection. This property marks a
distinguishing difference of the proposed detection scheme to the
well-established model- and observer-based fault detection framework.

\bigskip

It is worth emphasising that in a non-Euclidean space, the relation with the
above orthogonal projection as given in (\ref{eq2-21a}) does not hold. Our
intention of introducing the Bregman divergence as a distance measure is to
find an alternative solution. In this regard, it is of considerably
theoretic interest to study the projection $\hat{z}=\Sigma _{\mathcal{I}%
}\circ \left( D\Sigma _{\mathcal{I}}\right) ^{T}\left( z\right) $ from the
viewpoint of information geometry. To this end, the projection theorem \cite%
{Amari2016} is applied.

\bigskip

Consider the Bregman divergence from the process data $z$ to the system
image manifold $\mathcal{I}_{\Sigma }$ derived by $D_{H^{\times }},$%
\begin{equation}
D_{H^{\times }}\left[ z:\mathcal{I}_{\Sigma }\right] =\min_{z_{0}\in 
\mathcal{I}_{\Sigma }}D_{H^{\times }}\left[ z:z_{0}\right] .
\end{equation}%
The following theorem provides us with the interpretation that the
projection $\hat{z}$ delivered by the Hamiltonian system (\ref{eq3-1}) is a
geodesic projection.

\begin{Theo}
\label{Theo3-4}Consider the normalised SIR (\ref{eq2-2}) and the Bregman
divergence $D_{H^{\times }}\left[ z:\mathcal{I}_{\Sigma }\right] .$ The
projection $\hat{z}$ delivered by the corresponding Hamiltonian system (\ref%
{eq3-1}) is a geodesic projection of $z$ onto $\mathcal{I}_{\Sigma }.$
Moreover, the Bregman divergence $D_{H^{\times }}\left[ z:\mathcal{I}%
_{\Sigma }\right] $ is given by%
\begin{equation}
D_{H^{\times }}\left[ z:\mathcal{I}_{\Sigma }\right] =\frac{1}{2}z^{T}z+%
\frac{1}{2}\hat{z}^{T}\hat{z}-z^{T}\hat{z}.  \label{eq3-12a}
\end{equation}
\end{Theo}

\begin{proof}
It follows from Theorem \ref{Theo2-2} that the geodesic projection of $z$
onto $\mathcal{I}_{\Sigma }$ is the vector $z^{\ast }$ that minimises $%
D_{H^{\times }}\left[ z:z_{0}\right] ,\forall z_{0}\in \mathcal{I}_{\Sigma },
$ i.e. 
\begin{equation*}
D_{H^{\times }}\left[ z:\mathcal{I}_{\Sigma }\right] =\min_{z_{0}\in 
\mathcal{I}_{\Sigma }}D_{H^{\times }}\left[ z:z_{0}\right] =D_{H^{\times }}%
\left[ z:z^{\ast }\right] .
\end{equation*}%
Accordingly, it is to prove 
\begin{equation}
\hat{z}=z^{\ast }=\arg \min_{z_{0}\in \mathcal{I}_{\Sigma }}D_{H^{\times }}%
\left[ z:z_{0}\right] .  \label{eq3-7}
\end{equation}%
To this end, consider, following (\ref{eq2-22}) and noticing $\forall
z_{0}\in \mathcal{I}_{\Sigma },z_{0}^{\times }=\hat{z}_{0}=z_{0},$ 
\begin{equation*}
D_{H^{\times }}\left[ z:z_{0}\right] =H^{\times }\left( z\right) +H\left(
z_{0}^{\times }\right) -z^{T}z_{0}^{\times }=H^{\times }\left( z\right)
+H\left( z_{0}\right) -z^{T}z_{0}.
\end{equation*}%
It yields, by Theorems \ref{Theo3-1} and \ref{Theo3-2}, 
\begin{equation}
H\left( z_{0}\right) =\frac{1}{2}z_{0}^{T}z_{0}\Longrightarrow D_{H^{\times
}}\left[ z:z_{0}\right] =\frac{1}{2}z^{T}z+\frac{1}{2}%
z_{0}^{T}z_{0}-z^{T}z_{0}.  \label{eq3-7c}
\end{equation}%
Hence, the optimisation problem becomes 
\begin{equation*}
\min_{z_{0}\in \mathcal{I}_{\Sigma }}\left( \frac{1}{2}%
z_{0}^{T}z_{0}-z^{T}z_{0}\right) .
\end{equation*}%
Recall that $\forall z_{0}\in \mathcal{I}_{\Sigma }$%
\begin{gather}
D_{I}(x)D_{I}^{T}(x)z_{0}+c_{I}(x)+D_{I}(x)B_{I}^{T}(x)P_{x}\left( x\right)
=z_{0},  \label{eq3-7a} \\
D_{I}^{T}(x)D_{I}(x)=I,\left( c_{I}(x)+D_{I}(x)B_{I}^{T}(x)P_{x}\left(
x\right) \right) ^{T}D_{I}(x)=0,  \label{eq3-7b}
\end{gather}%
which yields 
\begin{gather*}
z_{0}^{T}z_{0}=z_{0}^{T}D_{I}(x)D_{I}^{T}(x)z_{0}+\left(
c_{I}(x)+D_{I}(x)B_{I}^{T}(x)P_{x}\left( x\right) \right) ^{T}\left(
c_{I}(x)+D_{I}(x)B_{I}^{T}(x)P_{x}\left( x\right) \right) , \\
z^{T}z_{0}=z^{T}\left( c_{I}(x)+D_{I}(x)B_{I}^{T}(x)P_{x}\left( x\right)
+D_{I}(x)D_{I}^{T}(x)z_{0}\right) .
\end{gather*}%
Thus, it holds 
\begin{equation*}
\min_{z_{0}\in \mathcal{I}_{\Sigma }}\left( \frac{1}{2}%
z_{0}^{T}z_{0}-z^{T}z_{0}\right) \Longleftrightarrow \min_{z_{0}\in \mathcal{%
I}_{\Sigma }}\left( \frac{1}{2}%
z_{0}^{T}D_{I}(x)D_{I}^{T}(x)z_{0}-z^{T}D_{I}(x)D_{I}^{T}(x)z_{0}\right) .
\end{equation*}%
Solving 
\begin{equation*}
\frac{\partial }{\partial z_{0}}\left( \frac{1}{2}%
z_{0}^{T}D_{I}(x)D_{I}^{T}(x)z_{0}-z^{T}D_{I}(x)D_{I}^{T}(x)z_{0}\right) =0
\end{equation*}%
gives 
\begin{equation}
D_{I}(x)D_{I}^{T}(x)z^{\ast }=D_{I}(x)D_{I}^{T}(x)z,  \label{eq3-8}
\end{equation}%
where $z^{\ast }$ represents the optimal solution that belongs to $\mathcal{I%
}_{\Sigma }.$ It is known that 
\begin{equation*}
\hat{z}=c_{I}(x)+D_{I}(x)B_{I}^{T}(x)P_{x}\left( x\right)
+D_{I}(x)D_{I}^{T}(x)z\in \mathcal{I}_{\Sigma }.
\end{equation*}%
It follows immediately from (\ref{eq3-7a})-(\ref{eq3-7b}) that 
\begin{equation*}
z^{\ast }=\hat{z}
\end{equation*}%
solves (\ref{eq3-8}). Finally, substituting $z_{0}$ in (\ref{eq3-7c}) by $%
\hat{z}$ leads to (\ref{eq3-12a}). The theorem is proved. 
\end{proof}

\bigskip

This result gives a nice interpretation of the projection $\hat{z}=\Sigma _{%
\mathcal{I}}\circ \left( D\Sigma _{\mathcal{I}}\right) ^{T}\left( z\right) .$
Thanks to the Bregman divergence $D_{H^{\times }}$ and the Legendre
transform, the data manifold is equipped with a dual Riemannian structure
and dually flat \cite{Amari2016}. In this context, the projection $\hat{z}$
is a geodesic projection of $z$ onto $\mathcal{I}_{\Sigma },$ and the
distance defined by the Bregman divergence $D_{H^{\times }}$ is at minimum.
This result can be interpreted as a generalisation of the orthogonal \
projection (\ref{eq2-0}) defined in Hilbert space.

\bigskip

As a summary of this subsection, it is concluded that the two major
objectives of our work, (i) system performance-oriented fault detection by
means of a Hamiltonian function induced Bregman divergence, (ii) detection
of deviations in a non-Euclidean space using a Bregman divergence as a
distance measure, have been satisfactorily achieved.

\subsection{Implementation issues\label{subsec3-3}}

Before the theoretical results presented in the previous subsections can be
applied to a reliable fault detection, two fault detection specified issues
should be clarified. They are,

\begin{itemize}
\item definition of a reliable (robust) evaluation function $J$ and, based
on it,

\item determination of a threshold $J_{th}$.
\end{itemize}

The first issue is of considerable importance to reduce false alarms. For
instance, the Bregman divergence $D_{H_{0}}\left[ \hat{z}:z\right] $ given
in (\ref{eq3-12}) is a time function. When it is directly used as an
evaluation function, detection decision will be made solely depending on the
data received at a single time instant. Due to possible transient
disturbances during system operations, such handling would probably result
in a flood of false alarms. Concerning the second issue, it is natural that,
under the detection logic, 
\begin{equation*}
\text{detection logic: }\left\{ 
\begin{array}{l}
J\leq J_{th}\Longrightarrow \text{fault-free} \\ 
J>J_{th}\Longrightarrow \text{faulty,}%
\end{array}%
\right.
\end{equation*}%
the selection of $J_{th}$ would have a crucial influence on the detection
performance. An integrated design of $\left( J,J_{th}\right) $ is state of
the art in model-based fault detection towards an optimal trade-off between
the false alarm rate and fault detection rate \cite{Ding2008,Ding2020}.
Below, these two issues are addressed jointly.

\bigskip

Let $\left[ t_{0},t_{1}\right] $ be the evaluation time window and define%
\begin{equation}
J=\int\limits_{t_{0}}^{t_{1}}D_{H_{0}}\left[ \hat{z}:z\right] d\tau
=\int\limits_{t_{0}}^{t_{1}}\left( H_{0}(\hat{z})-H_{0}(z)-\left( \frac{%
\partial H_{0}(z)}{\partial z}\right) ^{T}\left( \hat{z}-z\right) \right)
d\tau  \label{eq3-15}
\end{equation}%
as the evaluation function. In the literature, the divergence given by (\ref%
{eq3-15}) is called pointwise Bregman divergence \cite{FSG2008}. Consider
that, in real applications, digital data are collected and recorded. Suppose
that over the time interval $\left[ t_{0},t_{1}\right] ,$ $z\left(
k_{i}\right) ,i=1,\cdots ,M,$ are collected, where $k_{i}$ denotes the
sampling instant. For our purpose of fault detection, an approximation of
function (\ref{eq3-15}) is implemented using the available data $z\left(
k_{i}\right) $\ as follows, 
\begin{align}
J& =\frac{1}{M}\sum\limits_{i=1}^{M}D_{H_{0}}\left[ \hat{z}\left(
k_{i}\right) :z\left( k_{i}\right) \right]  \label{eq3-16} \\
& =\frac{1}{M}\sum\limits_{i=1}^{M}\left( H_{0}\left( \hat{z}\left(
k_{i}\right) \right) -H_{0}\left( z\left( k_{i}\right) \right) -\left( \frac{%
\partial H_{0}\left( z\left( k_{i}\right) \right) }{\partial z\left(
k_{i}\right) }\right) ^{T}\left( \hat{z}\left( k_{i}\right) -z\left(
k_{i}\right) \right) \right)  \notag \\
& =\frac{1}{M}\sum\limits_{i=1}^{M}\left( H_{0}\left( z\left( k_{i}\right)
\right) -H_{0}\left( \hat{z}\left( k_{i}\right) \right) -z^{T}\left(
k_{i}\right) \hat{z}\left( k_{i}\right) \right)  \notag \\
& =\frac{1}{M}\sum\limits_{i=1}^{M}\left( H_{0}\left( z\left( k_{i}\right)
\right) -H\left( \hat{z}\left( k_{i}\right) \right) \right) .  \notag
\end{align}%
It is of interest to note that the above evaluation function $J$ is
equivalent to the Bregman divergence defined by%
\begin{align}
D_{H_{0}}\left[ \hat{z}_{M}:z_{M}\right] & :=H_{0,M}\left( \hat{z}%
_{M}\right) -H_{0,M}\left( z_{M}\right) -\left( \frac{\partial H_{0,M}\left(
z_{M}\right) }{\partial z_{M}}\right) ^{T}\left( \hat{z}_{M}-z_{M}\right) ,
\label{eq3-17} \\
H_{0,M}\left( z_{M}\right) & =:\frac{1}{M}\sum\limits_{i=1}^{M}H_{0}\left(
z\left( k_{i}\right) \right) =\frac{1}{2}z_{M}^{T}z_{M},H_{0,M}\left( \hat{z}%
_{M}\right) =\frac{1}{2}\hat{z}_{M}^{T}\hat{z}_{M},  \label{eq3-18} \\
z_{M}& :=\left[ 
\begin{array}{c}
\frac{z\left( k_{1}\right) }{\sqrt{M}} \\ 
\vdots \\ 
\frac{z\left( k_{M}\right) }{\sqrt{M}}%
\end{array}%
\right] \in \mathbb{R}^{M\left( m+p\right) },\hat{z}_{M}:=\left[ 
\begin{array}{c}
\frac{\hat{z}\left( k_{1}\right) }{\sqrt{M}} \\ 
\vdots \\ 
\frac{\hat{z}\left( k_{M}\right) }{\sqrt{M}}%
\end{array}%
\right] \in \mathbb{R}^{M\left( m+p\right) },  \notag
\end{align}%
as proved in the following theorem.

\begin{Theo}
\label{Theo3-5}Given the normalised SIR (\ref{eq2-2}), the corresponding
Hamiltonian system (\ref{eq3-1}) and Hamiltonian function (\ref{eq3-4}), and
suppose that $z\left( k_{i}\right) ,i=1,\cdots ,M,$ are available. Then, the
evaluation function (\ref{eq3-16}) is equivalent to the Bregman divergence (%
\ref{eq3-17}) with the generating function (\ref{eq3-18}) and can be written
as%
\begin{align}
J& =D_{H_{0}}\left[ \hat{z}_{M}:z_{M}\right] =\frac{1}{2}z_{M}^{T}z_{M}+%
\frac{1}{2}\hat{z}_{M}^{T}\hat{z}_{M}-\hat{z}_{M}^{T}z_{M}  \label{eq3-19} \\
& =H_{0,M}\left( z_{M}\right) -H_{M}\left( z_{M}\right) ,  \label{eq3-19a} \\
H_{M}\left( z_{M}\right) & =\frac{1}{M}\sum\limits_{i=1}^{M}H\left( z\left(
k_{i}\right) \right) =\frac{1}{M}\sum\limits_{i=1}^{M}\left( \hat{z}%
^{T}\left( k_{i}\right) z\left( k_{i}\right) -H^{\times }\left( z\left(
k_{i}\right) \right) \right) .  \label{eq3-19b}
\end{align}
\end{Theo}

\begin{proof}
It is apparent that by definition%
\begin{align}
J& =\frac{1}{M}\sum\limits_{i=1}^{M}\left( H_{0}\left( \hat{z}\left(
k_{i}\right) \right) -H_{0}\left( z\left( k_{i}\right) \right) -\left( \frac{%
\partial H_{0}\left( z\left( k_{i}\right) \right) }{\partial z\left(
k_{i}\right) }\right) ^{T}\left( \hat{z}\left( k_{i}\right) -z\left(
k_{i}\right) \right) \right)   \notag \\
& =\frac{1}{2}z_{M}^{T}z_{M}+\frac{1}{2}\hat{z}_{M}^{T}\hat{z}%
_{M}-z_{M}^{T}\left( \hat{z}_{M}-z_{M}\right)   \notag \\
& =H_{0,M}\left( \hat{z}_{M}\right) -H_{0,M}\left( z_{M}\right) -\left( 
\frac{\partial H_{0,M}\left( z_{M}\right) }{\partial z_{M}}\right)
^{T}\left( \hat{z}_{M}-z_{M}\right) ,  \notag
\end{align}%
which proves (\ref{eq3-19}). Concerning (\ref{eq3-19a}) and (\ref{eq3-19b}),
notice that $\left( z_{M},\hat{z}_{M}\right) $ builds a dual pair of
Legendre transform with%
\begin{gather}
z_{M}=\frac{\partial H_{M}^{\times }\left( \hat{z}_{M}\right) }{\partial 
\hat{z}_{M}}=\left[ 
\begin{array}{c}
\frac{\partial H^{\times }\left( \hat{z}\left( k_{1}\right) \right) }{%
\partial \hat{z}\left( k_{1}\right) }\frac{1}{\sqrt{M}} \\ 
\vdots  \\ 
\frac{\partial H^{\times }\left( \hat{z}\left( k_{M}\right) \right) }{%
\partial \hat{z}\left( k_{M}\right) }\frac{1}{\sqrt{M}}%
\end{array}%
\right] =\left[ 
\begin{array}{c}
\frac{z\left( k_{1}\right) }{\sqrt{M}} \\ 
\vdots  \\ 
\frac{z\left( k_{M}\right) }{\sqrt{M}}%
\end{array}%
\right] ,\hat{z}_{M}=\frac{\partial H_{M}\left( z_{M}\right) }{\partial z_{M}%
},  \label{eq3-20a} \\
H_{M}\left( z_{M}\right) =\hat{z}_{M}^{T}z_{M}-H_{M}^{\times }\left( \hat{z}%
_{M}\right) =\frac{1}{M}\sum\limits_{i=1}^{M}\left( \hat{z}^{T}\left(
k_{i}\right) z\left( k_{i}\right) -H^{\times }\left( z\left( k_{i}\right)
\right) \right) ,  \label{eq3-20b}
\end{gather}%
where 
\begin{equation*}
H_{M}^{\times }\left( \hat{z}_{M}\right) =\frac{1}{M}\sum%
\limits_{i=1}^{M}H^{\times }\left( z\left( k_{i}\right) \right) .
\end{equation*}%
Consequently, the Bregman divergence (\ref{eq3-17}) can be written as%
\begin{equation*}
D_{H_{0}}\left[ \hat{z}_{M}:z_{M}\right] =H_{0,M}\left( z_{M}\right)
+H_{M}^{\times }\left( \hat{z}_{M}\right) -z_{M}^{T}\hat{z}%
_{M}=H_{0,M}\left( z_{M}\right) -H_{M}\left( z_{M}\right) 
\end{equation*}%
Hence, (\ref{eq3-19a}) and (\ref{eq3-19b}) are proved. 
\end{proof}

Theorem \ref{Theo3-5} and its proof not only reveal that the evaluation
function $J$ is a Bregman divergence, but also pose the basis of a potential
framework of data-driven fault detection. The centerpiece of this framework
is the Bregman divergence $D_{H_{0}}\left[ \hat{z}_{M}:z_{M}\right] $ and
the Legendre transform (\ref{eq3-20a})-(\ref{eq3-20b}), which would enable
to endow the data set 
\begin{equation*}
\mathcal{M}=\left\{ \left[ 
\begin{array}{c}
z\left( k_{1}\right) \\ 
\vdots \\ 
z\left( k_{M}\right)%
\end{array}%
\right] \in \mathcal{L}_{2},z\left( k_{i}\right) =\left[ 
\begin{array}{c}
u\left( k_{i}\right) \\ 
y\left( k_{i}\right)%
\end{array}%
\right] ,i=1,\cdots ,M\right\}
\end{equation*}%
with a dual Riemannian structure and thus a projection-based fault
detection. This issue is a part of our future work and will not be discussed
in the draft.

\bigskip

Before introducing the threshold setting scheme, we would like to place a
greater emphasis on the basic task of one-class fault detection. Simply
speaking, it is to detect anomalies arising during normal system operations.
The challenge is a reliable detection of those incipient anomalies as early
as possible in spite of possibly existing uncertainties in the system. Now,
we are in the position to study the threshold setting issue on the
assumption of the evaluation function $J$ given by (\ref{eq3-16}).

\bigskip

Recall that the Bregman divergence $D_{H_{0}}\left[ \hat{z}_{M}:z_{M}\right] 
$ gives the difference between the nominal Hamiltonian function $%
H_{0,M}\left( z_{M}\right) $ as a reference and the real Hamiltonian
function $H_{M}\left( z_{M}\right) $ over the time interval $\left[
t_{0},t_{1}\right] .$ During the ideal operation without faults and
uncertainty, 
\begin{equation*}
z_{M}=\hat{z}_{M},H_{M}^{\times }\left( z_{M}\right) =H_{0,M}\left(
z_{M}\right) =\frac{1}{2}z_{M}^{T}z_{M}\Longrightarrow J=D_{H_{0}}\left[ 
\hat{z}_{M}:z_{M}\right] =0.
\end{equation*}%
The term $z_{M}^{T}\hat{z}_{M}$ describes the influence of anomalies and
uncertainties causing $z_{M}\neq \hat{z}_{M}.$ Since%
\begin{equation*}
\hat{z}_{M}=\frac{\partial H_{M}\left( z_{M}\right) }{\partial z_{M}}
\end{equation*}%
gives the directional derivative of $H_{M}\left( z_{M}\right) $ in the
direction $z_{M},$ it is at maximum, when the direction $\hat{z}_{M}$
coincides with the gradient of $H_{M}\left( z_{M}\right) .$ In other words,
during the nominal operation, $H_{M}\left( z_{M}\right) $ is at maximum,%
\begin{equation*}
H_{M}\left( z_{M}\right) =\left( \frac{\partial H_{M}\left( z_{M}\right) }{%
\partial z_{M}}\right) ^{T}z_{M}-\frac{1}{2}\left( \frac{\partial
H_{M}\left( z_{M}\right) }{\partial z_{M}}\right) ^{T}\frac{\partial
H_{M}\left( z_{M}\right) }{\partial z_{M}}=\frac{1}{2}z_{M}^{T}z_{M}.
\end{equation*}
As uncertainties or faults emerge, it becomes smaller. Let $\gamma \in \left[
0.5,1\right] $ be the tolerant limit so that it is accepted as fault-free
operation when 
\begin{equation*}
\gamma \leq \frac{H_{M}\left( z_{M}\right) }{\frac{1}{2}z_{M}^{T}z_{M}}=%
\frac{2\hat{z}_{M}^{T}z_{M}-\hat{z}_{M}^{T}\hat{z}_{M}}{z_{M}^{T}z_{M}}.
\end{equation*}%
Accordingly, when 
\begin{equation*}
\frac{D_{H_{0}}\left[ \hat{z}_{M}:z_{M}\right] }{\frac{1}{2}z_{M}^{T}z_{M}}%
=1-\frac{H_{M}\left( z_{M}\right) }{\frac{1}{2}z_{M}^{T}z_{M}}=1-\frac{2\hat{%
z}_{M}^{T}z_{M}-\hat{z}_{M}^{T}\hat{z}_{M}}{z_{M}^{T}z_{M}}\leq 1-\gamma ,
\end{equation*}%
the system operation is identified as fault-free. In this context, the
threshold is set to be%
\begin{equation}
J_{th}=\frac{1}{2}\left( 1-\gamma \right) z_{M}^{T}z_{M}.
\end{equation}%
Here, $\gamma $ is assumed to be known or determined from the real
fault-free operation or simulation data, say e.g. $\gamma =0.95.$

\bigskip

At the end of this section, the online detection algorithm is summarised as
follows:

\begin{itemize}
\item Collection of data $z\left( k_{i}\right) $ and computation of $\hat{z}%
\left( k_{i}\right) ,i=1,\cdots ,M,$ according to (\ref{eq3-1a});

\item Computation of the evaluation function $J$ according to (\ref{eq3-19});

\item Detection decision logic:%
\begin{equation*}
\left\{ 
\begin{array}{l}
J\leq J_{th}\Longrightarrow \text{ fault-free} \\ 
J>J_{th}\Longrightarrow \text{ faulty.}%
\end{array}%
\right.
\end{equation*}
\end{itemize}

\section{An SKR-based fault detection scheme, and uncertainty/fault
estimation}

Thanks to the dual relation of SKR and SIR, most of the schemes and ideas of
our study in this section on the application of the normalised SKR-based
projection to fault detection are analogue to the ones presented in the
previous section. Hence, our descriptions and discussions on the related
issues will be shortened. On the other hand, intensive endeavours will be
made to highlight the projection-based estimation of system uncertainties
and faults aiming at insightfully understanding and interpreting the
normalised SKR-based projection.

\subsection{On the uncertainty corrupted in process data and its estimation}

In this subsection, we take a closer look at system $\left( D\Sigma _{%
\mathcal{K}}\right) ^{T}\circ \Sigma _{\mathcal{K}}$ and delineate its use
for estimating uncertainties/faults corrupted in the process data $\left(
u,y\right) .$ We begin with the Hamiltonian extension of $\Sigma _{\mathcal{K%
}},$%
\begin{equation}
\left\{ 
\begin{array}{l}
\dot{\hat{x}}=a_{K}(\hat{x})+B_{K}(\hat{x})z,z:=\left[ 
\begin{array}{c}
u \\ 
y%
\end{array}%
\right] \\ 
\dot{\lambda}=-\left( \frac{\partial a_{K}(\hat{x})}{\partial \hat{x}}+\frac{%
\partial B_{K}(\hat{x})}{\partial \hat{x}}z\right) ^{T}\lambda -\left( \frac{%
\partial c_{K}(\hat{x})}{\partial \hat{x}}+\frac{\partial D_{K}(\hat{x})}{%
\partial \hat{x}}\right) ^{T}z_{a} \\ 
r_{y}=c_{K}(\hat{x})+D_{K}(\hat{x})z \\ 
r_{y,a}=B_{K}^{T}(\hat{x})\lambda +D_{K}^{T}(\hat{x})z_{a},z_{a}\in \mathbb{R%
}^{m},r_{y,a}\in \mathbb{\ R}^{m+p}.%
\end{array}%
\right.  \label{eq4-0}
\end{equation}%
Connecting $r_{y}$ and $z_{a},$ and defining the Hamiltonian function, 
\begin{equation}
H\left( \hat{x},\lambda ,z\right) =\frac{1}{2}\left( c_{K}(\hat{x})+D_{K}(%
\hat{x})z\right) ^{T}\left( c_{K}(\hat{x})+D_{K}(\hat{x})z\right) +\lambda
^{T}\left( a_{K}(\hat{x})+B_{K}(\hat{x})z\right) ,  \label{eq4-1}
\end{equation}%
give the Hamiltonian system $\left( D\Sigma _{\mathcal{K}}\right) ^{T}\circ
\Sigma _{\mathcal{K}},$%
\begin{equation}
\left( D\Sigma _{\mathcal{K}}\right) ^{T}\circ \Sigma _{\mathcal{K}}:%
\begin{cases}
\dot{\hat{x}}=\frac{\partial H\left( \hat{x},\lambda ,z\right) }{\partial
\lambda } \\ 
\dot{\lambda}=-\frac{\partial H\left( \hat{x},\lambda ,z\right) }{\partial 
\hat{x}} \\ 
r_{y,a}=\frac{\partial H\left( \hat{x},\lambda ,z\right) }{\partial z}.%
\end{cases}
\label{eq4-2}
\end{equation}%
System $\left( D\Sigma _{\mathcal{K}}\right) ^{T}\circ \Sigma _{\mathcal{K}}$
is an $\mathcal{L}_{2}$ $\rightarrow $ $\mathcal{L}_{2}$ mapping which can
be interpreted as an estimation of the uncertainty corrupted in the data $z,$
as will be discussed in the sequel. For this reason, we adopt the notation 
\begin{equation*}
\hat{z}_{\Delta }:=r_{y,a}=\left( D\Sigma _{\mathcal{K}}\right) ^{T}\circ
\Sigma _{\mathcal{K}}\left( z\right) .
\end{equation*}%
Figure \ref{Fig3-2} sketches the configuration of $\left( D\Sigma _{\mathcal{%
K}}\right) ^{T}\circ \Sigma _{\mathcal{K}}$ schematically. 
\begin{figure}[h]
\centering\includegraphics[width=9cm,height=2cm]{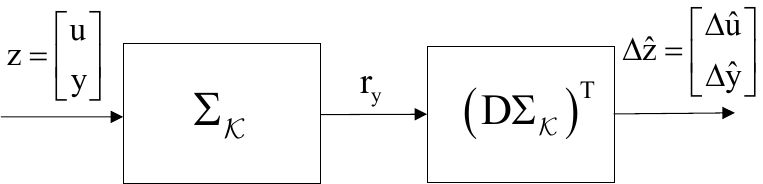}
\caption{SKR-based projection system $\left( D\Sigma _{\mathcal{K}}\right)
^{T}\circ \Sigma _{\mathcal{K}}$}
\label{Fig3-2}
\end{figure}

\subsubsection{Co-inner and uncertainties}

Recall that the normalised SKR is a co-inner. That means, by definition, the
system%
\begin{equation*}
r_{y}=\Sigma _{\mathcal{K}}\circ \left( D\Sigma _{\mathcal{K}}\right)
^{T}\left( z_{a}\right) ,
\end{equation*}%
configured by connecting $z$ and $r_{y,a}$ in the Hamiltonian extension of $%
\Sigma _{\mathcal{K}}$ (\ref{eq4-0}), is of the lossless property, namely, 
\begin{eqnarray*}
r_{y} &=&z_{a}, \\
\exists V(x) &\geq &0,V(0)=0,\text{ s.t. }V\left( x\left( t_{2}\right)
\right) -V\left( x\left( t_{1}\right) \right) =\frac{1}{2}%
\int_{t_{1}}^{t_{2}}\left( r_{y,a}^{T}r_{y,a}-r_{y}^{T}r_{y}\right) d\tau .
\end{eqnarray*}%
It leads to 
\begin{equation}
r_{y,a}=\left( D\Sigma _{\mathcal{K}}\right) ^{T}\left( z_{a}\right) =\left(
D\Sigma _{\mathcal{K}}\right) ^{T}\left( r_{y}\right) .  \label{eq4-0a}
\end{equation}%
Since the residual vector $r_{y}$ contains informations about uncertainties
in the system, relation (\ref{eq4-0a}) indicates that the output of the
system $\left( D\Sigma _{\mathcal{K}}\right) ^{T}\left( r_{y}\right) $ is of
the character of uncertainties. This observation draws our attention on the
(normalised) SKR and system uncertainties.

\subsubsection{System kernel and uncertainty models}

In advanced control and model-based fault diagnosis theories, handling of
system uncertainties is an essential \ issue. There are numerous ways to
model various types of model uncertainties and, based on them, the
influences (often in form of an upper-bound) of system uncertainties on the
system output or the residual signal can be estimated and used, e.g. for
robust controller design \cite%
{Francis87,Vinnicombe-book,Zhou96,vanderSchaftbook} or threshold setting in
a fault detection scheme \cite{Ding2008,Ding2020}. Below, we introduce a
model that describes uncertainties corrupted in the system data $\left(
u,y\right) $. This model underlines the insightful interpretation of the
projection $\hat{z}_{\Delta }=\left( D\Sigma _{\mathcal{K}}\right) ^{T}\circ
\Sigma _{\mathcal{K}}\left( z\right) $ as an estimator for
uncertainties/faults.

\bigskip

To begin with, a modified form of the SKR $\Sigma _{\mathcal{K}}$ (\ref%
{eq2-3}), 
\begin{equation}
\left\{ 
\begin{array}{l}
\dot{\hat{x}}=a_{K}(\hat{x})+B_{K}(\hat{x})\left[ 
\begin{array}{c}
u \\ 
y%
\end{array}%
\right] \\ 
r_{y}=c_{K}(\hat{x})+D_{K}(\hat{x})\left[ 
\begin{array}{c}
u \\ 
y%
\end{array}%
\right] \\ 
y=\hat{y}+W^{-1}(\hat{x})r_{y},\hat{y}=c(\hat{x})+D(\hat{x})u,%
\end{array}%
\right.  \label{eq4-10}
\end{equation}%
is introduced. Model (\ref{eq4-10}) describes the I/O-dynamics of the system
(\ref{eq2-1}), and thus is called kernel-based I/O model. Apparently, it
models the I/O dynamics both in the nominal and uncertain/faulty operations,
and is independent of all possible types of system uncertainties. A
plausible explanation of this distinguishing capability is the fact that
system uncertainties are implicitly corrupted in the process data $\left(
u,y\right) $ and represented by the residual vector $r_{y}.$ This raises the
question of whether it is possible to (i) model the uncertainties corrupted
in $\left( u,y\right) $, and (ii) estimate them. An apparent answer in case
of a linear system, as shown below, motivates our work.

\bigskip

\begin{example}
Let 
\begin{equation}
\left[ 
\begin{array}{c}
u \\ 
y%
\end{array}%
\right] =\left[ 
\begin{array}{c}
u_{0} \\ 
y_{0}%
\end{array}%
\right] +\left[ 
\begin{array}{c}
\Delta u \\ 
\Delta y%
\end{array}%
\right] ,  \label{eq4-11}
\end{equation}%
where $\left[ 
\begin{array}{c}
u_{0} \\ 
y_{0}%
\end{array}%
\right] \in \mathcal{I}_{\Sigma }$ denotes the nominal operation, and $\left[
\begin{array}{c}
\Delta u \\ 
\Delta y%
\end{array}%
\right] $ represents the uncertainties, for instance, caused by unknown
input vectors in the system. As described in Section \ref{sec2-1}, we have%
\begin{align*}
\left\{ 
\begin{array}{l}
\dot{\hat{x}}=A_{L}\hat{x}+\left[ 
\begin{array}{cc}
B_{L} & \text{ }-L%
\end{array}%
\right] \left[ 
\begin{array}{c}
u \\ 
y%
\end{array}%
\right]  \\ 
r_{y}=-W\left( C\hat{x}+\left[ 
\begin{array}{cc}
-D & \text{ }I%
\end{array}%
\right] \left[ 
\begin{array}{c}
u \\ 
y%
\end{array}%
\right] \right) 
\end{array}%
\right. & =\left\{ 
\begin{array}{l}
\dot{\hat{x}}_{\Delta }=A_{L}\hat{x}_{\Delta }+\left[ 
\begin{array}{cc}
B_{L} & \text{ }-L%
\end{array}%
\right] \left[ 
\begin{array}{c}
\Delta u \\ 
\Delta y%
\end{array}%
\right]  \\ 
r_{y}=-W\left( C\hat{x}_{\Delta }+\left[ 
\begin{array}{cc}
-D & \text{ }I%
\end{array}%
\right] \left[ 
\begin{array}{c}
\Delta u \\ 
\Delta y%
\end{array}%
\right] \right) ,%
\end{array}%
\right.  \\
A_{L}& =A-LC,B_{L}=B-LD,
\end{align*}%
since it holds $K_{G}I_{G}=0.$ Here, the state variables $\hat{x}$ and $\hat{%
x}_{\Delta }$ are generally different, and the two systems are equal in the
context that they deliver the identical output. 
\end{example}

Although for nonlinear systems the data corrupted with uncertainties cannot
be written in the additive form like (\ref{eq4-11}), this example inspires
the introduction of the following uncertainty model.

\begin{Def}
Given the normalised SKR (\ref{eq2-3}), the state space representation of $%
\Delta \Sigma _{\mathcal{K}},$ 
\begin{equation}
\Delta \Sigma _{\mathcal{K}}:\left\{ 
\begin{array}{l}
\dot{\hat{x}}_{\Delta }=a_{K}(\hat{x}_{\Delta })+B_{K}(\hat{x}_{\Delta })%
\left[ 
\begin{array}{c}
\Delta u \\ 
\Delta y%
\end{array}%
\right] \\ 
r_{y}=c_{K}(\hat{x}_{\Delta })+D_{K}(\hat{x}_{\Delta })\left[ 
\begin{array}{c}
\Delta u \\ 
\Delta y%
\end{array}%
\right] ,%
\end{array}%
\right. \hat{x}_{\Delta }(0)=0,\left[ 
\begin{array}{c}
\Delta u \\ 
\Delta y%
\end{array}%
\right] =:\Delta z\in \mathcal{L}_{2},  \label{eq4-12}
\end{equation}%
subject to 
\begin{equation}
\left\{ 
\begin{array}{l}
\dot{\hat{x}}_{\Delta }=a_{K}(\hat{x}_{\Delta })+B_{K}(\hat{x}_{\Delta })%
\left[ 
\begin{array}{c}
\Delta u \\ 
\Delta y%
\end{array}%
\right] \\ 
r_{y}=c_{K}(\hat{x}_{\Delta })+D_{K}(\hat{x}_{\Delta })\left[ 
\begin{array}{c}
\Delta u \\ 
\Delta y%
\end{array}%
\right]%
\end{array}%
\right. =\left\{ 
\begin{array}{l}
\dot{\hat{x}}=a_{K}(\hat{x})+B_{K}(\hat{x})\left[ 
\begin{array}{c}
u \\ 
y%
\end{array}%
\right] \\ 
r_{y}=c_{K}(\hat{x})+D_{K}(\hat{x})\left[ 
\begin{array}{c}
u \\ 
y%
\end{array}%
\right] ,%
\end{array}%
\right.  \label{eq4-13}
\end{equation}%
is called uncertainty model of the SKR with $\Delta z$ denoting the
uncertainties corrupted in the data $\left( u,y\right) .$
\end{Def}

The uncertainty model (\ref{eq4-12}) and the condition (\ref{eq4-13}) imply
that,

\begin{itemize}
\item on the one hand, the residual vector $r_{y}$ exclusively contains the
information about uncertainties in the system,

\item on the other hand, $r_{y}$ can be equivalently generated by the SKR (%
\ref{eq2-3}) driven by $\left( u,y\right) .$
\end{itemize}

Although $\Delta z$ is in general unknown, its value can be specified in two
(extreme) cases, which play an essential role in our subsequent study:

\begin{itemize}
\item corresponding to the nominal dynamics,%
\begin{equation}
\forall \left[ 
\begin{array}{c}
u \\ 
y%
\end{array}%
\right] \in \mathcal{I}_{\Sigma },\left[ 
\begin{array}{c}
\Delta u \\ 
\Delta y%
\end{array}%
\right] =0;  \label{eq4-21a}
\end{equation}

\item when it holds%
\begin{align}
\hat{z}_{\Delta }& =\left( D\Sigma _{\mathcal{K}}\right) ^{T}\circ \Sigma _{%
\mathcal{K}}\left( \left[ 
\begin{array}{c}
u \\ 
y%
\end{array}%
\right] \right) =\left( D\Sigma _{\mathcal{K}}\right) ^{T}\circ \Sigma _{%
\mathcal{K}}\left( \left[ 
\begin{array}{c}
\Delta u \\ 
\Delta y%
\end{array}%
\right] \right)  \label{eq4-21b} \\
& =\left[ 
\begin{array}{c}
u \\ 
y%
\end{array}%
\right] =\left[ 
\begin{array}{c}
\Delta u \\ 
\Delta y%
\end{array}%
\right] .  \notag
\end{align}
\end{itemize}

The first case is true due to the relation (\ref{eq2-4}) and reflects the
system response during nominal operations.\ The second one means, the data $%
\left( u,y\right) $ solely contains uncertainties, which, despite being
rarely the case in real operations, characterises a crucial data set, and
induces us to introduce the concept of manifold of uncertain data.

\begin{Def}
\label{Def4-2}Given the normalised SKR (\ref{eq2-3}), manifold $\mathcal{U}%
_{\Sigma }\subset \mathcal{L}_{2},$%
\begin{equation}
\mathcal{U}_{\Sigma }=\left\{ \left[ 
\begin{array}{c}
u \\ 
y%
\end{array}%
\right] ,\left[ 
\begin{array}{c}
u \\ 
y%
\end{array}%
\right] =\left( D\Sigma _{\mathcal{K}}\right) ^{T}(r),r\in \mathcal{L}%
_{2}\right\} ,
\end{equation}%
is called manifold of uncertain data.
\end{Def}

It is apparent that $\forall \left[ 
\begin{array}{c}
u \\ 
y%
\end{array}%
\right] \in \mathcal{U}_{\Sigma },$%
\begin{equation*}
\hat{z}_{\Delta }=\left( D\Sigma _{\mathcal{K}}\right) ^{T}\circ \Sigma _{%
\mathcal{K}}\left( \left[ 
\begin{array}{c}
u \\ 
y%
\end{array}%
\right] \right) =\left[ 
\begin{array}{c}
u \\ 
y%
\end{array}%
\right] ,
\end{equation*}%
i.e. the relation (\ref{eq4-21b}) holds. Thus, the data in $\mathcal{U}%
_{\Sigma }$ solely contain uncertainties. Note that 
\begin{equation}
\forall \left[ 
\begin{array}{c}
u \\ 
y%
\end{array}%
\right] \in \mathcal{L}_{2},\left[ 
\begin{array}{c}
\Delta \hat{u} \\ 
\Delta \hat{y}%
\end{array}%
\right] =\left( D\Sigma _{\mathcal{K}}\right) ^{T}\circ \Sigma _{\mathcal{K}%
}\left( \left[ 
\begin{array}{c}
u \\ 
y%
\end{array}%
\right] \right) \in \mathcal{U}_{\Sigma }  \label{eq4-14}
\end{equation}%
indicates that the operator $\left( D\Sigma _{\mathcal{K}}\right) ^{T}\circ
\Sigma _{\mathcal{K}}$ projects the data $\left( u,y\right) $ to the
manifold $\mathcal{U}_{\Sigma }.$ Observe\ that, as a result of the
uncertainty model (\ref{eq4-12}), 
\begin{equation}
\left[ 
\begin{array}{c}
\hat{u}_{\Delta } \\ 
\hat{y}_{\Delta }%
\end{array}%
\right] =\left( D\Sigma _{\mathcal{K}}\right) ^{T}\circ \Sigma _{\mathcal{K}%
}\left( \left[ 
\begin{array}{c}
u \\ 
y%
\end{array}%
\right] \right) =\left( D\Sigma _{\mathcal{K}}\right) ^{T}\circ \Sigma _{%
\mathcal{K}}\left( \left[ 
\begin{array}{c}
\Delta u \\ 
\Delta y%
\end{array}%
\right] \right) .  \label{eq4-16}
\end{equation}%
That is, $\left[ 
\begin{array}{c}
\hat{u}_{\Delta } \\ 
\hat{y}_{\Delta }%
\end{array}%
\right] $ is a projection of the uncertainty $\left[ 
\begin{array}{c}
\Delta u \\ 
\Delta y%
\end{array}%
\right] $ corrupted in the system data and thus serves as an estimate of $%
\left[ 
\begin{array}{c}
\Delta u \\ 
\Delta y%
\end{array}%
\right] $ in the context of the uncertainty model (\ref{eq4-12}). It can be,
for instance, applied for re-constructing the detected faults or examining
the data quality. To this end, we will study the normalised SKR and the
associated Hamiltonian projection system, 
\begin{equation}
\hat{z}_{\Delta }=\left[ 
\begin{array}{c}
\hat{u}_{\Delta } \\ 
\hat{y}_{\Delta }%
\end{array}%
\right] =\left( D\Sigma _{\mathcal{K}}\right) ^{T}\circ \Sigma _{\mathcal{K}%
}\left( \left[ 
\begin{array}{c}
u \\ 
y%
\end{array}%
\right] \right) ,  \label{eq4-18}
\end{equation}%
in the next subsection.

\begin{Rem}
We would like to emphasise that $\hat{z}_{\Delta }$ is an estimation of the
uncertainty $\left[ 
\begin{array}{c}
\Delta u \\ 
\Delta y%
\end{array}%
\right] $ corrupted in the system data. And the Hamiltonian projection
system (\ref{eq4-18}) is a projection of the data $\left( u,y\right) $ onto
the manifold of uncertain data $\mathcal{U}_{\Sigma }$.
\end{Rem}

\subsection{Normalised SKR and the associated Hamiltonian projection system}

We begin with the state space description of system (\ref{eq4-18}). Let 
\begin{equation}
H\left( \hat{x},\lambda ,z\right) =\frac{1}{2}\left( c_{K}(\hat{x})+D_{K}(%
\hat{x})z\right) ^{T}\left( c_{K}(\hat{x})+D_{K}(\hat{x})z\right) +\lambda
^{T}\left( a_{K}(\hat{x})+B_{K}(\hat{x})z\right)  \label{eq4-19}
\end{equation}%
be the Hamiltonian function, where, for the sake of simplicity, notation $%
\hat{x}$ instead of $\hat{x}_{\Delta }$ is abused. Accordingly, we have%
\begin{equation}
\left( D\Sigma _{\mathcal{K}}\right) ^{T}\circ \Sigma _{\mathcal{K}}:%
\begin{cases}
\dot{\hat{x}}=\frac{\partial H\left( \hat{x},\lambda ,z\right) }{\partial
\lambda } \\ 
\dot{\lambda}=-\frac{\partial H\left( \hat{x},\lambda ,z\right) }{\partial 
\hat{x}} \\ 
\hat{z}_{\Delta }=\frac{\partial H\left( \hat{x},\lambda ,z\right) }{%
\partial z}.%
\end{cases}
\label{eq4-20}
\end{equation}%
By some routine calculations, the Hamiltonian function (\ref{eq4-19}) is
further written into 
\begin{align}
H\left( \hat{x},\lambda ,z\right) & =\frac{1}{2}r_{y}^{T}r_{y}+\lambda ^{T}%
\dot{\hat{x}}  \notag \\
& =\hat{z}_{\Delta }^{T}\hat{z}_{\Delta }+\frac{1}{2}c_{K}^{T}(\hat{x})c_{K}(%
\hat{x})-\frac{1}{2}z^{T}D_{K}^{T}(\hat{x})D_{K}(\hat{x})z+\lambda ^{T}a_{K}(%
\hat{x})  \notag \\
& =\hat{z}_{\Delta }^{T}z-\frac{1}{2}\hat{z}_{\Delta }^{T}\hat{z}_{\Delta
}+c_{K}^{T}(\hat{x})c_{K}(\hat{x})+\lambda ^{T}B_{K}(\hat{x})D_{K}^{T}(\hat{x%
})c_{K}(\hat{x})  \notag \\
& +\frac{1}{2}\lambda ^{T}B_{K}(\hat{x})B_{K}^{T}(\hat{x})\lambda +\lambda
^{T}a_{K}(\hat{x}).  \label{eq4-2a}
\end{align}%
The calculation of the last equation follows from the relations 
\begin{equation*}
D_{K}(\hat{x})D_{K}^{T}(\hat{x})=I,D_{K}^{T}(\hat{x})D_{K}(\hat{x})D_{K}^{T}(%
\hat{x})D_{K}(\hat{x})=D_{K}^{T}(\hat{x})D_{K}(\hat{x}),
\end{equation*}%
which are true by setting $W(\hat{x})=\left( I+D(\hat{x})D^{T}(\hat{x}%
)\right) ^{-1/2}$ and are assumed in our subsequent study. The following
theorem gives the Hamiltonian functions $H\left( \hat{x},\lambda ,z\right) $
and $H^{\times }\left( \hat{x},\lambda ,\hat{z}_{\Delta }\right) $
corresponding to the normalised SKR.

\begin{Theo}
\label{Theo4-1}Given the normalised SKR (\ref{eq2-3}) and the corresponding
Hamiltonian system (\ref{eq4-20}), then 
\begin{align}
H\left( \hat{x},\lambda ,z\right) & =\hat{z}_{\Delta }^{T}z-\frac{1}{2}\hat{z%
}_{\Delta }^{T}\hat{z}_{\Delta },  \label{eq4-3} \\
H^{\times }\left( \hat{x},\lambda ,\hat{z}_{\Delta }\right) & =\frac{1}{2}%
\hat{z}_{\Delta }^{T}\hat{z}_{\Delta }.  \label{eq4-4}
\end{align}
\end{Theo}

\begin{proof}
It follows from Lemma \ref{Lemma2-1} and Theorem \ref{Theo2-1} that for $%
\lambda =V_{\hat{x}}\left( \hat{x}\right) $ with $V\left( \hat{x}\right) $
solving (\ref{eq2-15c}) 
\begin{gather*}
c_{K}^{T}(\hat{x})c_{K}(\hat{x})+V_{\hat{x}}\left( \hat{x}\right) B_{K}(\hat{%
x})D_{K}^{T}(\hat{x})c_{K}(\hat{x})=0, \\
V_{\hat{x}}^{T}\left( \hat{x}\right) a_{K}(\hat{x})+\frac{1}{2}V_{\hat{x}%
}^{T}\left( \hat{x}\right) B_{K}(\hat{x})B_{K}^{T}(\hat{x})V_{\hat{x}}\left( 
\hat{x}\right) =0 \\
\Longrightarrow H^{\times }\left( \hat{x},\lambda ,\hat{z}_{\Delta }\right) =%
\frac{1}{2}\hat{z}_{\Delta }^{T}\hat{z}_{\Delta },H\left( \hat{x},\lambda
,z\right) =\hat{z}_{\Delta }^{T}z-H^{\times }\left( \hat{x},\lambda ,\hat{z}%
_{\Delta }\right) =\hat{z}_{\Delta }^{T}z-\frac{1}{2}\hat{z}_{\Delta }^{T}%
\hat{z}_{\Delta }.
\end{gather*}
\end{proof}

For our study, we are interested in the Hamiltonian functions corresponding
to the two extreme cases: (i) $z\in \mathcal{I}_{\Sigma },$ (ii) $\forall
z\in \mathcal{U}_{\Sigma }.$ The following corollary gives the solutions.

\begin{Corol}
\label{col4-1}Given the normalised SKR (\ref{eq2-3}) and the corresponding
Hamiltonian system (\ref{eq4-20}), then%
\begin{gather}
\forall z=\left[ 
\begin{array}{c}
u \\ 
y%
\end{array}%
\right] \in \mathcal{I}_{\Sigma },H^{\times }\left( x,\lambda ,\hat{z}%
_{\Delta }\right) =H\left( x,\lambda ,z\right) =0,  \label{eq4-21} \\
\forall z=\left[ 
\begin{array}{c}
u \\ 
y%
\end{array}%
\right] \in \mathcal{U}_{\Sigma },H\left( \hat{x},\lambda ,z\right)
=H^{\times }\left( \hat{x},\lambda ,\hat{z}_{\Delta }\right) =\frac{1}{2}%
\hat{z}_{\Delta }^{T}\hat{z}_{\Delta }=\frac{1}{2}z^{T}z.  \label{eq4-22}
\end{gather}
\end{Corol}

\begin{proof}
It follows from (\ref{eq2-4}) (and (\ref{eq4-21a})) that%
\begin{equation*}
\forall z=\left[ 
\begin{array}{c}
u \\ 
y%
\end{array}%
\right] \in \mathcal{I}_{\Sigma },\left[ 
\begin{array}{c}
\Delta u \\ 
\Delta y%
\end{array}%
\right] =0,r_{y}=0\Longrightarrow \hat{z}_{\Delta }=0.
\end{equation*}%
Hence, according to Theorem \ref{Theo4-1},%
\begin{equation*}
H^{\times }\left( x,\lambda ,\hat{z}_{\Delta }\right) =H\left( x,\lambda
,z\right) =0.
\end{equation*}%
By Definition \ref{Def4-2}, 
\begin{equation*}
\forall z\in \mathcal{U}_{\Sigma },\hat{z}_{\Delta }=\left( D\Sigma _{%
\mathcal{K}}\right) ^{T}\circ \Sigma _{\mathcal{K}}\left( z\right) =\left(
D\Sigma _{\mathcal{K}}\right) ^{T}\circ \Sigma _{\mathcal{K}}\left( \Delta
z\right) =z=\Delta z.
\end{equation*}%
Thus, it follows from Theorem \ref{Theo4-1} that 
\begin{equation*}
H\left( \hat{x},\lambda ,z\right) =H^{\times }\left( \hat{x},\lambda ,\hat{z}%
_{\Delta }\right) =\frac{1}{2}\hat{z}_{\Delta }^{T}\hat{z}_{\Delta }=\frac{1%
}{2}z^{T}z.
\end{equation*}
\end{proof}

\bigskip

Corollary \ref{col4-1} describes two special operation situations, in which 
\begin{equation*}
H^{\times }\left( x,\lambda ,\hat{z}_{\Delta }\right) =H\left( x,\lambda
,z\right) .
\end{equation*}%
The first one given by (\ref{eq4-21}) corresponds to the nominal dynamics,
while the second one with (\ref{eq4-22}) reflects the situation, when the
data exclusively consist of uncertainties.

\subsection{Applications to fault detection and uncertainty estimation}

Based on the results presented in the previous subsections, Bregman
divergence-based fault detection schemes as\ well as a further study on
uncertainty estimation are addressed in this subsection.

\subsubsection{Bregman divergence-based fault detection}

To begin with, we discuss about the application of the Bregman divergence
from the estimated uncertainty in the data, $\hat{z}_{\Delta },$ to the
manifold $\mathcal{I}_{\Sigma }$ for the purpose of fault detection.
Consider the uncertainty model (\ref{eq4-12}) and recall that 
\begin{equation}
\forall \left[ 
\begin{array}{c}
u \\ 
y%
\end{array}%
\right] \in \mathcal{I}_{\Sigma },\left[ 
\begin{array}{c}
\Delta u \\ 
\Delta y%
\end{array}%
\right] =\left[ 
\begin{array}{c}
\hat{u}_{\Delta } \\ 
\hat{y}_{\Delta }%
\end{array}%
\right] =0.  \label{eq4-6a}
\end{equation}%
To detect possible uncertainties in the data $\left( u,y\right) ,$ the
Hamiltonian function 
\begin{equation*}
H^{\times }\left( \hat{z}_{\Delta }\right) :=H^{\times }\left( \hat{x}%
,\lambda ,\hat{z}_{\Delta }\right) =\frac{1}{2}\hat{z}_{\Delta }^{T}\hat{z}%
_{\Delta }
\end{equation*}%
can be used. Note that the Bregman divergence from $\hat{z}_{\Delta }$ to
the manifold $\mathcal{I}_{\Sigma }$ derived from $H^{\times }\left( \hat{z}%
_{\Delta }\right) $ is given by, due to (\ref{eq4-6a}), 
\begin{equation}
D_{H^{\times }}\left[ \hat{z}_{\Delta }:0\right] =H^{\times }\left( \hat{z}%
_{\Delta }\right) =\frac{1}{2}\hat{z}_{\Delta }^{T}\hat{z}_{\Delta }.
\label{eq4-6}
\end{equation}%
It is of interest to notice that the Bregman divergence (\ref{eq4-6}) is
true for all $\left( u,y\right) $ belonging to $\mathcal{I}_{\Sigma }.$
Moreover, $\hat{z}_{\Delta }=\left( D\Sigma _{\mathcal{K}}\right) ^{T}\circ
\Sigma _{\mathcal{K}}\left( z\right) .$ In this sense, $D_{H^{\times }}\left[
\hat{z}_{\Delta }:0\right] $ can be interpreted as a divergence from the
uncertainty in the data to $\mathcal{I}_{\Sigma }.$

\bigskip

Analogue to the discussion on the implementation issues in Subsection \ref%
{subsec3-3}, we shortly summarise the major steps to apply the Bregman
divergence (\ref{eq4-6}) for a practical fault detection. Let $z\left(
k_{i}\right) ,i=1,\cdots ,M,$ be the data collected over the time interval $%
\left[ t_{0},t_{1}\right] $ and 
\begin{align}
J& =\frac{1}{M}\sum\limits_{i=1}^{M}D_{H^{\times }}\left[ \hat{z}_{\Delta
}\left( k_{i}\right) :0\right] =\frac{1}{2M}\sum\limits_{i=1}^{M}\hat{z}%
_{\Delta }^{T}\left( k_{i}\right) \hat{z}_{\Delta }\left( k_{i}\right) =%
\frac{1}{2}\hat{z}_{\Delta M}^{T}\hat{z}_{\Delta M},  \label{eq4-8} \\
\hat{z}_{\Delta M}& =\left[ 
\begin{array}{c}
\frac{\hat{z}_{\Delta }\left( k_{1}\right) }{\sqrt{M}} \\ 
\vdots \\ 
\frac{\hat{z}_{\Delta }\left( k_{M}\right) }{\sqrt{M}}%
\end{array}%
\right] \in \mathbb{R}^{M\left( m+p\right) },z_{M}=\left[ 
\begin{array}{c}
\frac{z\left( k_{1}\right) }{\sqrt{M}} \\ 
\vdots \\ 
\frac{z\left( k_{M}\right) }{\sqrt{M}}%
\end{array}%
\right] \in \mathbb{R}^{M\left( m+p\right) },  \notag
\end{align}%
be the corresponding evaluation function. It is straightforward that the
evaluation function (\ref{eq4-8}) is a Bregman divergence with the
generating function $\frac{1}{2}\hat{z}_{\Delta M}^{T}\hat{z}_{\Delta M},$
namely%
\begin{equation*}
J=D_{H^{\times }}\left[ z_{M}:0\right] =\frac{1}{2}\hat{z}_{\Delta M}^{T}%
\hat{z}_{\Delta M}.
\end{equation*}%
To set the threshold, consider the ratio%
\begin{equation*}
\frac{\hat{z}_{\Delta M}^{T}\hat{z}_{\Delta M}}{z_{M}^{T}z_{M}}
\end{equation*}%
that gives an estimate of the relative change in the process data $z_{M}.$
Let $\alpha \in \left( 0,1\right) $ be the tolerant limit so that it is
accepted as fault-free operation when 
\begin{equation*}
\frac{\hat{z}_{\Delta M}^{T}\hat{z}_{\Delta M}}{z_{M}^{T}z_{M}}\leq \alpha .
\end{equation*}%
Accordingly, the threshold is set to be%
\begin{equation}
J_{th}=\frac{\alpha }{2}z_{M}^{T}z_{M}.  \label{eq4-9}
\end{equation}%
Below is the online detection algorithm:

\begin{itemize}
\item Collection of data $z\left( k_{i}\right) $ and computation of $\hat{z}%
\left( k_{i}\right) ,i=1,\cdots ,M,$ according to (\ref{eq4-2});

\item Computation of the evaluation function $J$ according to (\ref{eq4-8});

\item Detection decision logic:%
\begin{equation*}
\left\{ 
\begin{array}{l}
J\leq J_{th}\Longrightarrow \text{ fault-free} \\ 
J>J_{th}\Longrightarrow \text{ faulty.}%
\end{array}%
\right.
\end{equation*}
\end{itemize}

\subsubsection{On the projection-based estimation\label{subsec4-3-2}}

Having studied the use of the projection, $\hat{z}_{\Delta }=\left( D\Sigma
_{\mathcal{K}}\right) ^{T}\left( r_{y}\right) =\left( D\Sigma _{\mathcal{K}%
}\right) ^{T}\circ \Sigma _{\mathcal{K}}\left( z\right) ,$ for the fault
detection purpose, we now comprehend $\hat{z}_{\Delta }$ as an estimate of
uncertainties corrupted in the data. To this end, two issues are addressed:
(i) examining $\hat{z}_{\Delta }$ as a projection of the system data on the
uncertain data manifold $\mathcal{U}_{\Sigma },$ from the viewpoint of
information geometry, and (ii) considering $\hat{z}_{\Delta }$ as an
estimate of $\Delta z$ subject to the uncertainty model (\ref{eq4-12}) and
examining the corresponding residual $r_{y}.$

\bigskip

The following theorem is a dual result of Theorem \ref{Theo3-4} and gives $%
\hat{z}_{\Delta }$ an insightful interpretation from the information
geometric viewpoint.

\begin{Theo}
\label{Theo4-3}Consider the normalised SKR (\ref{eq2-3}) and the Bregman
divergence $D_{H^{\times }}\left[ z:\mathcal{U}_{\Sigma }\right] ,$ 
\begin{align*}
D_{H^{\times }}\left[ z:\mathcal{U}_{\Sigma }\right] & =\min_{\Delta
z_{0}\in \mathcal{U}_{\Sigma }}D_{H^{\times }}\left[ z:\Delta z_{0}\right] ,
\\
D_{H^{\times }}\left[ z:\Delta z_{0}\right] & :=H^{\times }\left( z\right)
-H^{\times }(\Delta z_{0})-\left( \frac{\partial H^{\times }(\Delta z_{0}))}{%
\partial \Delta z_{0}}\right) ^{T}\left( z-\Delta z_{0}\right) , \\
H^{\times }\left( z\right) & =\hat{z}_{\Delta }^{T}z-H\left( \hat{x},\lambda
,z\right) =\frac{1}{2}\hat{z}_{\Delta }^{T}\hat{z}_{\Delta }.
\end{align*}%
The projection $\hat{z}_{\Delta }$ delivered by the corresponding
Hamiltonian system (\ref{eq4-20}) is a geodesic projection of $z$ onto $%
\mathcal{U}_{\Sigma }.$
\end{Theo}

\begin{proof}
The proof is similar to the one of Theorem \ref{Theo3-4}. Thus, only key
steps are described. According to Theorem \ref{Theo2-2}, the geodesic
projection of $z$ onto $\mathcal{U}_{\Sigma }$ is the vector $\Delta z^{\ast
}$ leading to 
\begin{equation*}
\min_{\Delta z_{0}\in \mathcal{U}_{\Sigma }}D_{H^{\times }}\left[ z:\Delta
z_{0}\right] =D_{H^{\times }}\left[ z:\Delta z^{\ast }\right] .
\end{equation*}%
Following (\ref{eq2-22}) and noticing 
\begin{equation*}
\forall \Delta z_{0}\in \mathcal{U}_{\Sigma },\Delta z_{0}^{\times }=\hat{z}%
_{\Delta 0}=\left( D\Sigma _{\mathcal{K}}\right) ^{T}\circ \Sigma _{\mathcal{%
K}}\left( \Delta z_{0}\right) =\Delta z_{0},
\end{equation*}%
it turns out%
\begin{equation*}
D_{H^{\times }}\left[ z:\Delta z_{0}\right] =H^{\times }\left( z\right)
+H\left( \Delta z_{0}^{\times }\right) -z^{T}\Delta z_{0}^{\times
}=H^{\times }\left( z\right) +H\left( \Delta z_{0}\right) -z^{T}\Delta z_{0}.
\end{equation*}%
It yields, by Theorem \ref{Theo4-1} and Corollary \ref{col4-1}, 
\begin{equation*}
H\left( \Delta z_{0}\right) =\frac{1}{2}\Delta z_{0}^{T}\Delta
z_{0}\Longrightarrow D_{H^{\times }}\left[ z:\Delta z_{0}\right] =H^{\times
}\left( z\right) +\frac{1}{2}\Delta z_{0}^{T}\Delta z_{0}-z^{T}\Delta z_{0},
\end{equation*}%
which leads to 
\begin{equation*}
\min_{\Delta z_{0}\in \mathcal{U}_{\Sigma }}D_{H^{\times }}\left[ z:\Delta
z_{0}\right] \Longleftrightarrow \min_{\Delta z_{0}\in \mathcal{U}_{\Sigma
}}\left( \frac{1}{2}\Delta z_{0}^{T}\Delta z_{0}-z^{T}\Delta z_{0}\right) .
\end{equation*}%
Since $\forall \Delta z_{0}\in \mathcal{U}_{\Sigma },$%
\begin{align*}
\Delta z_{0}& =B_{K}^{T}(\hat{x})V_{\hat{x}}\left( \hat{x}\right) +D_{K}^{T}(%
\hat{x})c_{K}(\hat{x})+D_{K}^{T}(\hat{x})D_{K}(\hat{x})\Delta z_{0}, \\
D_{K}(\hat{x})D_{K}^{T}(\hat{x})& =I,\left( B_{K}^{T}(\hat{x})V_{\hat{x}%
}\left( \hat{x}\right) +D_{K}^{T}(\hat{x})c_{K}(\hat{x})\right)
^{T}D_{K}^{T}(\hat{x})=0,
\end{align*}%
it holds 
\begin{gather*}
\Delta z_{0}^{T}\Delta z_{0}=\Delta z_{0}^{T}D_{K}^{T}(\hat{x})D_{K}(\hat{x}%
)\Delta z_{0} \\
+\left( B_{K}^{T}(\hat{x})V_{\hat{x}}\left( \hat{x}\right) +D_{K}^{T}(\hat{x}%
)c_{K}(\hat{x})\right) ^{T}\left( B_{K}^{T}(\hat{x})V_{\hat{x}}\left( \hat{x}%
\right) +D_{K}^{T}(\hat{x})c_{K}(\hat{x})\right) , \\
z^{T}\Delta z_{0}=z^{T}\left( B_{K}^{T}(\hat{x})V_{\hat{x}}\left( \hat{x}%
\right) +D_{K}^{T}(\hat{x})c_{K}(\hat{x})+D_{K}^{T}(\hat{x})D_{K}(\hat{x}%
)\Delta z_{0}\right) \Longrightarrow  \\
\min_{\Delta z_{0}\in \mathcal{U}_{\Sigma }}\left( \frac{1}{2}\Delta
z_{0}^{T}\Delta z_{0}-z^{T}\Delta z_{0}\right) \Longleftrightarrow  \\
\min_{\Delta z_{0}\in \mathcal{U}_{\Sigma }}\left( \frac{1}{2}\Delta
z_{0}^{T}D_{K}^{T}(\hat{x})D_{K}(\hat{x})\Delta z_{0}-z^{T}D_{K}^{T}(\hat{x}%
)D_{K}(\hat{x})\Delta z_{0}\right) .
\end{gather*}%
Solving 
\begin{equation*}
\frac{\partial }{\partial \Delta z_{0}}\left( \frac{1}{2}\Delta
z_{0}^{T}D_{K}^{T}(\hat{x})D_{K}(\hat{x})\Delta z_{0}-z^{T}D_{K}^{T}(\hat{x}%
)D_{K}(\hat{x})\Delta z_{0}\right) =0
\end{equation*}%
gives 
\begin{equation}
D_{K}^{T}(\hat{x})D_{K}(\hat{x})\Delta z^{\ast }=D_{K}^{T}(\hat{x})D_{K}(%
\hat{x})z  \label{eq4-8a}
\end{equation}%
where $\Delta z^{\ast }$ represents the optimal solution belonging to $%
\mathcal{U}_{\Sigma }.$ It is obvious that 
\begin{equation*}
\hat{z}_{\Delta }=B_{K}^{T}(\hat{x})V_{\hat{x}}\left( \hat{x}\right)
+D_{K}^{T}(\hat{x})c_{K}(\hat{x})+D_{K}^{T}(\hat{x})D_{K}(\hat{x})z\in 
\mathcal{U}_{\Sigma }
\end{equation*}%
solves (\ref{eq4-8a}), which proves $\Delta z^{\ast }=\hat{z}_{\Delta }.$ 
\end{proof}

\bigskip

Theorem \ref{Theo4-3} underlines $\hat{z}_{\Delta }$ as an optimal estimate
for $\Delta z$ in the context of a geodesic projection of $z$ onto $\mathcal{%
U}_{\Sigma }.$ Here, the manifold $\mathcal{U}_{\Sigma }$ is equipped with a
dual Riemannian structure by means of the Bregman divergence $D_{H^{\times
}} $ and it is dually flat \cite{Amari2016}.

\bigskip

Next, consider the SKR (\ref{eq2-3}) and the estimator 
\begin{align}
\hat{z}_{\Delta }& =\left[ 
\begin{array}{c}
\hat{u}_{\Delta } \\ 
\hat{y}_{\Delta }%
\end{array}%
\right] =\left( D\Sigma _{\mathcal{K}}\right) ^{T}\left( r_{y}\right) ,
\label{eq4-15} \\
\left( D\Sigma _{\mathcal{K}}\right) ^{T}& :\left\{ 
\begin{array}{l}
\dot{\lambda}=-\left( \frac{\partial a_{K}(\hat{x})}{\partial \hat{x}}+\frac{%
\partial B_{K}(\hat{x})}{\partial \hat{x}}z\right) ^{T}\lambda -\left( \frac{%
\partial c_{K}(\hat{x})}{\partial \hat{x}}+\frac{\partial D_{K}(\hat{x})}{%
\partial \hat{x}}\right) ^{T}r_{y} \\ 
\left[ 
\begin{array}{c}
\hat{u}_{\Delta } \\ 
\hat{y}_{\Delta }%
\end{array}%
\right] =B_{K}^{T}(\hat{x})\lambda +D_{K}^{T}(\hat{x})r_{y}.%
\end{array}%
\right.
\end{align}%
To guarantee that the delivered estimate $\hat{z}_{\Delta }$ is subject to
the uncertainty model (\ref{eq4-12}), it is required that%
\begin{equation}
\Sigma _{\mathcal{K}}\left( \left[ 
\begin{array}{c}
\hat{u}_{\Delta } \\ 
\hat{y}_{\Delta }%
\end{array}%
\right] \right) =r_{y}\Longrightarrow r_{y}=\Sigma _{\mathcal{K}}\circ
\left( D\Sigma _{\mathcal{K}}\right) ^{T}\left( r_{y}\right) .
\label{eq4-23}
\end{equation}%
It is apparent that the normalised SKR, as a co-inner, satisfies the
condition (\ref{eq4-23}). Below, it is delineated that a normalised SKR
based estimator (\ref{eq4-15}) delivers an LS estimation in the context that
it solves the optimisation problem 
\begin{align}
& \min_{L\left( \hat{x}\right) }\frac{1}{2}\int\limits_{0}^{\infty
}r_{y}^{T}r_{y}dt  \label{eq4-24} \\
\text{s.t. }\dot{\hat{x}}& =a_{K}(\hat{x})+B_{K}(\hat{x})\hat{z}_{\Delta }. 
\notag
\end{align}%
To this end, an arbitrary SKR with 
\begin{equation*}
W(\hat{x})=\left( I+D(\hat{x})D^{T}(\hat{x})\right) ^{-1/2}\Longrightarrow
D_{K}(\hat{x})D_{K}^{T}(\hat{x})=I
\end{equation*}%
is considered.

\bigskip

Observe that the condition (\ref{eq4-23}) implies 
\begin{equation*}
r_{y}=c_{K}(\hat{x})+D_{K}(\hat{x})B_{K}^{T}(\hat{x})\lambda
+r_{y}\Longleftrightarrow c_{K}(\hat{x})+D_{K}(\hat{x})B_{K}^{T}(\hat{x}%
)\lambda =0.
\end{equation*}%
Moreover, it holds%
\begin{gather*}
\frac{1}{2}r_{y}^{T}r_{y}+\lambda ^{T}\left( a_{K}(\hat{x})+B_{K}(\hat{x}%
)B_{K}^{T}(\hat{x})\lambda +B_{K}(\hat{x})D_{K}^{T}(\hat{x})r_{y}\right) \\
=\frac{1}{2}\left( B_{K}^{T}(\hat{x})\lambda +D_{K}^{T}(\hat{x})r_{y}\right)
^{T}\left( B_{K}^{T}(\hat{x})\lambda +D_{K}^{T}(\hat{x})r_{y}\right)
+\lambda ^{T}\left( a_{K}(\hat{x})+\frac{1}{2}B_{K}(\hat{x})B_{K}^{T}(\hat{x}%
)\lambda \right) \\
=\frac{1}{2}\hat{z}_{\Delta }^{T}\hat{z}_{\Delta }+\lambda ^{T}\left( a_{K}(%
\hat{x})+\frac{1}{2}B_{K}(\hat{x})B_{K}^{T}(\hat{x})\lambda \right) .
\end{gather*}%
Now, let 
\begin{align*}
\bar{H}\left( \hat{x},\lambda ,\hat{z}_{\Delta }\right) & :=\frac{1}{2}\hat{z%
}_{\Delta }^{T}\hat{z}_{\Delta }+\lambda ^{T}\left( a_{K}(\hat{x})+\frac{1}{2%
}B_{K}(\hat{x})B_{K}^{T}(\hat{x})\lambda \right) \\
& =\frac{1}{2}r_{y}^{T}r_{y}+\lambda ^{T}\left( a_{K}(\hat{x})+B_{K}(\hat{x})%
\hat{z}_{\Delta }\right) ,
\end{align*}%
and calculate 
\begin{equation*}
\frac{\partial }{\partial L\left( \hat{x}\right) }\bar{H}\left( \hat{x}%
,\lambda ,\hat{z}_{\Delta }\right) =\frac{\partial }{\partial L\left( \hat{x}%
\right) }\left( 
\begin{array}{c}
\frac{1}{2}\lambda ^{T}L\left( \hat{x}\right) \left( I+D\left( \hat{x}%
\right) D^{T}\left( \hat{x}\right) \right) L^{T}\left( \hat{x}\right) \lambda
\\ 
-\lambda ^{T}L\left( \hat{x}\right) \left( c\left( \hat{x}\right) +D\left( 
\hat{x}\right) B^{T}(\hat{x})\lambda \right)%
\end{array}%
\right) .
\end{equation*}%
It is apparent that $\lambda $ and $L^{\ast }\left( \hat{x}\right) $ adopted
in the normalised SKR and satisfying (\ref{eq2-15c})-(\ref{eq2-15d}) solve
the equation%
\begin{equation*}
\frac{\partial }{\partial L\left( \hat{x}\right) }\bar{H}\left( \hat{x}%
,\lambda ,\hat{z}_{\Delta }\right) =0.
\end{equation*}%
As a result, we obtain%
\begin{equation}
L^{\ast }\left( \hat{x}\right) =\arg \min_{L\left( \hat{x}\right) }\bar{H}%
\left( \hat{x},\lambda ,\hat{z}_{\Delta }\right) .  \label{eq4-17}
\end{equation}%
It is well-known from the calculus of variations that%
\begin{equation*}
\bar{H}\left( \hat{x},\lambda ,\hat{z}_{\Delta }\right) =\frac{1}{2}%
r_{y}^{T}r_{y}+\lambda ^{T}\left( a_{K}(\hat{x})+B_{K}(\hat{x})\hat{z}%
_{\Delta }\right)
\end{equation*}%
can be interpreted as the Hamiltonian function of the optimisation problem (%
\ref{eq4-24}), and the optimal solution given by solving (\ref{eq4-17}). In
this regard, the uncertainty estimate $\hat{z}_{\Delta }$ delivered by the
estimator (\ref{eq4-15}) is an LS estimate subject to the uncertainty model (%
\ref{eq4-12}) with $r_{y}$ as its output whose $\mathcal{L}_{2}$-norm is at
a minimum.

\section{Conclusions}

In this draft, performance-oriented one-class fault detection and estimation
in nonlinear dynamic systems with uncertainties have been addressed. Bearing
in mind that the objective of this work is to establish a framework, in
which not only model-based but also data-driven and ML-based fault diagnosis
strategies can be uniformly handled, our work has followed a paradigm that
is different from the conventional control theoretically oriented
model-based framework. To be specific, our work is based on the SIR and SKR
model forms of dynamic systems, rather than on the well-established
input-output and the associated state space models. This change enables us
to model the nominal system dynamics as a lower-dimensional manifold
embedded in the process data space $\left( u,y\right) .$ The main idea
behind this handling is, along the line of data-driven and ML-based fault
detection schemes, to deal with fault detection as a classification problem
in the process data space. To this end, projection technique provides us
with a capable mathematical tool. In this regard, our work has focused on
projection-based fault detection and estimation in nonlinear dynamic systems.

\bigskip

Encouraged by our previous work on application of the orthogonal projection
technique to fault detection in LTI systems \cite{ding2022}, we have started
our effort with constructing projection systems. As an extension of the
well-established projections onto the image and kernel subspaces of LTI
systems \cite{Georgiou88,Vinnicombe-book,ding2022}, normalised SIR and SKR
of nonlinear systems have been adopted for our purpose. In order to address
nonlinear system specified issues, we have focused on the lossless
properties of the normalised SIR and SKR as inner and co-inner, and
introduced various configurations of Hamiltonian extensions \ \cite%
{Schaft-book1987} as Hamiltonian systems \cite%
{Scherpen1994,Ball1996,PS-AC2005}. They serve as the system theoretical tool
throughout our work.

\bigskip

The theoretical basis for a successful application of the orthogonal
projection technique to fault detection in LTI systems is the Pythagorean
Theorem. Concretely, the normalised SIR-based orthogonal projection of the
process data $\left( u,y\right) $ onto the image subspace in Hilbert space,
describing the nominal system dynamics, allows an orthogonal decomposition
of $\left( u,y\right) $ into an optimal estimate (projection) and the
residual. By means of the Pythagorean equation, a norm-based evaluation of $%
\left( u,y\right) $ leads to a unique separation of the nominal dynamic and
the dynamics corresponding to uncertainties/faults. As a result, a
classification and thus an optimal fault detection is achieved. For
nonlinear dynamic systems, (process) data processing should be performed in
non-Euclidean space. Consequently, the norm-based distance defined in
Hilbert space may not be compatible with the metric tensor of the
non-Euclidean space and hence is not suitable to measure the distance from a
data vector to the system image manifold. As one of\ the major contributions
of our work, we have proposed to use a Bregman divergence, a measure of
difference between two points in a space, to solve this problem. A Bregman
divergence is defined in terms of a convex function. When the system
performance function under consideration is convex, it is plain that we are
able to use the Bregman divergence defined by such a performance to achieve
performance-oriented fault detection. A further distinguishing contribution
of our work is the proposed scheme of combining the normalised SIR as
Hamiltonian systems and the Bregman divergence induced by the corresponding
Hamiltonian functions. This scheme not only enables to realise\ the
objective of performance-oriented fault detection, but also uncovers the
information geometric aspect of our work. It is known that the input and
output of a Hamiltonian system build a Legendre transform and there exists a
Hamiltonian system as the Legendre dual system. As a result, the computation
of the corresponding Bregman divergence can be expressed in terms of the
dual Hamiltonian functions. In particular, a Bregman divergence together
with a Legendre transform may induce a dually flat Riemannian manifold in
the data space that may be regarded as a dualistic extension of the
Euclidean space \cite{Amari2016,Nielsen2020}. In this context, the
Hamiltonian system based projection can be interpreted as a geodesic
projection. So far, this part of work is an integration of control theory
and information geometry, whose main results are summarised as follows:

\begin{itemize}
\item Construction of the normalised SIR-based projection system that
projects $\left( u,y\right) $ onto $\mathcal{I}_{\Sigma },$ the image
manifold describing the nominal system dynamics,

\item Determination of the Hamiltonian functions with respect to the nominal
dynamics,

\item Determination and computation of the Bregman divergence for
performance-oriented fault detection, and based on it,

\item Definition of the evaluation function over a time interval, which has
been proved to be a Bregman divergence as well, and the corresponding
threshold setting,

\item Demonstration of the interpretation that the projection onto the image
manifold is a geodesic projection.
\end{itemize}

As reported in our previous work \cite{ding2022}, for LTI systems the
orthogonal projection onto the image subspace is equivalent to a projection
onto the kernel subspace. Thanks to the Pythagorean equation, the\
projection-based fault detection can then be equivalently realised using an
observer-based residual generator. These conclusions and results have
motivated and inspired our study on the normalised SKR induced Hamiltonian
system $\left( D\Sigma _{\mathcal{K}}\right) ^{T}\circ \Sigma _{\mathcal{K}}$
as a projection. Since the SKR $\Sigma _{\mathcal{K}}$ is a residual
generator and the residual contains uncertainty information, we have firstly
studied $\left( D\Sigma _{\mathcal{K}}\right) ^{T}\circ \Sigma _{\mathcal{K}%
} $ and its co-inner property from the information (uncertainty) viewpoint.
To be specific, the uncertainty model (\ref{eq4-12}) and Definition \ref%
{Def4-2}, the manifold of uncertain data $\mathcal{U}_{\Sigma },$ have been
introduced. It has been proved that $\left( D\Sigma _{\mathcal{K}}\right)
^{T}\circ \Sigma _{\mathcal{K}}$ represents a projection of the process data
onto $\mathcal{U}_{\Sigma }.$ Analogue to our work in the first part, a
fault detection scheme has been proposed, consisting of

\begin{itemize}
\item design of $\left( D\Sigma _{\mathcal{K}}\right) ^{T}\circ \Sigma _{%
\mathcal{K}}$ and the corresponding dual Hamiltonian system,

\item determination of the Hamiltonian functions with respect to the two
extreme cases, $z\in \mathcal{I}_{\Sigma }$ and $z\in \mathcal{U}_{\Sigma },$

\item definition and computation of the Hamiltonian function induced Bregman
divergence and its use for fault detection,

\item analysis and interpretation of $\left( D\Sigma _{\mathcal{K}}\right)
^{T}\circ \Sigma _{\mathcal{K}}$ as a projection onto $\mathcal{U}_{\Sigma
}, $ and

\item study on the realisation issues.
\end{itemize}

The further contribution of our study on $\left( D\Sigma _{\mathcal{K}%
}\right) ^{T}\circ \Sigma _{\mathcal{K}}$ is the proof that the projection
is an LS estimation of the uncertainty corrupted in $\left( u,y\right) $ and
subject to the uncertainty model. An immediate application of this result is
fault estimation. In other words, variations in the process data caused by
faults, in particular those process faults, are modelled like uncertainties
by means of the uncertainty model and the manifold $\mathcal{U}_{\Sigma }$
represents the corresponding set \ in the data space. And $\left( D\Sigma _{%
\mathcal{K}}\right) ^{T}\circ \Sigma _{\mathcal{K}}$ serves as an optimal
(LS) estimator for such variations in the process data.

\bigskip

Our future work will concentrate on two aspects aiming at the overall goal
of our efforts to build a uniform framework. The first one is the work on an
ML-based realisation of the Hamiltonian projection systems. In this work, we
will focus on AE learning methods. In fact, both projection systems studied
in this draft, $\Sigma _{\mathcal{I}}\circ \left( D\Sigma _{\mathcal{I}%
}\right) ^{T}$ and $\left( D\Sigma _{\mathcal{K}}\right) ^{T}\circ \Sigma _{%
\mathcal{K}}$, are of the typical configuration of an AE, where the latent
variables are $v$ and $r_{y}$ respectively. In our previous work \cite%
{LDLCX2022}, we have proposed to train $\Sigma _{\mathcal{I}}\circ \left(
D\Sigma _{\mathcal{I}}\right) ^{T}$ as neural networks (NNs) in such a way
that $\Sigma _{\mathcal{I}}\circ \left( D\Sigma _{\mathcal{I}}\right) ^{T}$
is idempotent and of the lossless property. In our future work, the
corresponding Bregman divergences and Legendre transform will be integrated
into the learning process, and $\left( D\Sigma _{\mathcal{K}}\right)
^{T}\circ \Sigma _{\mathcal{K}}$ will be considered as well. This work can
be interpreted as control theoretically learning of AE for optimal and
data-driven fault diagnosis. The second aspect is the application of PINN
methods \cite{PINN2021}. In our theoretical work, we notice that solving
HJEs like (\ref{eq2-15a}) and (\ref{eq2-15c}) analytically is a hard work.
The PINN methods offer a capable tool to deal with such problems.

\bigskip

\textbf{Acknowledgement}: The authors are very grateful to Prof. Y. Yang for
the collaborative and valuable support.

\end{document}